\documentclass[acmtog]{acmart}
\usepackage{amsthm}
\usepackage{amsmath}
\usepackage{enumitem}
\usepackage{xcolor}
\usepackage{wrapfig}

\theoremstyle{definition}
\newtheorem{definition}{Definition}[section] 
\theoremstyle{definition}

\theoremstyle{definition}
\newtheorem{prerequisite}{Prerequisite}[section]
\theoremstyle{definition}
\newtheorem{lemma}{Lemma}[section]
\theoremstyle{definition}

\usepackage[ruled]{algorithm2e} % For algorithms

\SetAlFnt{\small}
\SetAlCapFnt{\small}
\SetAlCapNameFnt{\small}
\SetAlCapHSkip{0pt}

\setcopyright{acmlicensed}
\copyrightyear{2025}
\acmYear{2025}
\acmDOI{10.1145/3763355}
\acmJournal{TOG}

\begin{document}
\title{Can Any Model Be Fabricated? Inverse Operation Based Planning for Hybrid Additive–Subtractive Manufacturing}

%%%%%%%%%%%%%
%！！！ Each author must have a valid e-mail address.
%！！！ Each affiliation must have a valid country. (City and state / province is nice, too.)
%%%%%%%%%%%%%

% DO NOT ENTER AUTHOR INFORMATION FOR ANONYMOUS TECHNICAL PAPER SUBMISSIONS TO SIGGRAPH 2019!
\author{Yongxue Chen}
\authornotemark[1]
\orcid{0000-0001-6236-4158}
\affiliation{%
 \institution{The University of Manchester}
 \streetaddress{Oxford Rd}
 \city{Manchester}
 \postcode{M13 9PL}
 \country{United Kingdom}}

\author{Tao Liu}
\authornote{Joint first authors.}
\orcid{0000-0003-1016-4191}
\affiliation{%
 \institution{The University of Manchester}
 \streetaddress{Oxford Rd}
 \city{Manchester}
 \postcode{M13 9PL}
 \country{United Kingdom}}

\author{Yuming Huang}
\orcid{0000-0001-5900-2164}
\affiliation{%
 \institution{The University of Manchester}
 \streetaddress{Oxford Rd}
 \city{Manchester}
 \postcode{M13 9PL}
 \country{United Kingdom}}

\author{Weiming Wang}
\orcid{0000-0001-6289-0094}
\affiliation{%
 \institution{The University of Manchester}
 \streetaddress{Oxford Rd}
 \city{Manchester}
 \postcode{M13 9PL}
 \country{United Kingdom}}

\author{Tianyu Zhang}
\orcid{0000-0003-0372-0049}
\affiliation{%
 \institution{The University of Manchester}
 \streetaddress{Oxford Rd}
 \city{Manchester}
 \postcode{M13 9PL}
 \country{United Kingdom}}

\author{Kun Qian}
\orcid{0000-0002-8719-1537}
\affiliation{%
 \institution{The University of Manchester}
 \streetaddress{Oxford Rd}
 \city{Manchester}
 \postcode{M13 9PL}
 \country{United Kingdom}}

\author{Zikang Shi}
\orcid{0000-0002-2420-7722}
\affiliation{%
 \institution{The University of Manchester}
 \streetaddress{Oxford Rd}
 \city{Manchester}
 \postcode{M13 9PL}
 \country{United Kingdom}}

\author{Charlie C.L. Wang}
\orcid{0000-0003-4406-8480}
\authornote {Corresponding author: charlie.wang@machester.ac.uk (Charlie C.L. Wang).}
\affiliation{%
 \institution{The University of Manchester}
 \streetaddress{Oxford Rd}
 \city{Manchester}
 \postcode{M13 9PL}
 \country{United Kingdom}}
\email{charlie.wang@manchester.ac.uk}

\begin{abstract}
This paper presents a method for computing interleaved additive and subtractive manufacturing operations to fabricate models of arbitrary shapes. We solve the manufacturing planning problem by searching a sequence of inverse operations that progressively transform a target model into a null shape. Each inverse operation corresponds to either an additive or a subtractive step, ensuring both manufacturability and structural stability of intermediate shapes throughout the process. We theoretically prove that any model can be fabricated exactly using a sequence generated by our approach. To demonstrate the effectiveness of this method, we adopt a voxel-based implementation and develop a scalable algorithm that works on models represented by a large number of voxels. Our approach has been tested across a range of digital models and further validated through physical fabrication on a hybrid manufacturing system with automatic tool switching.
\end{abstract}

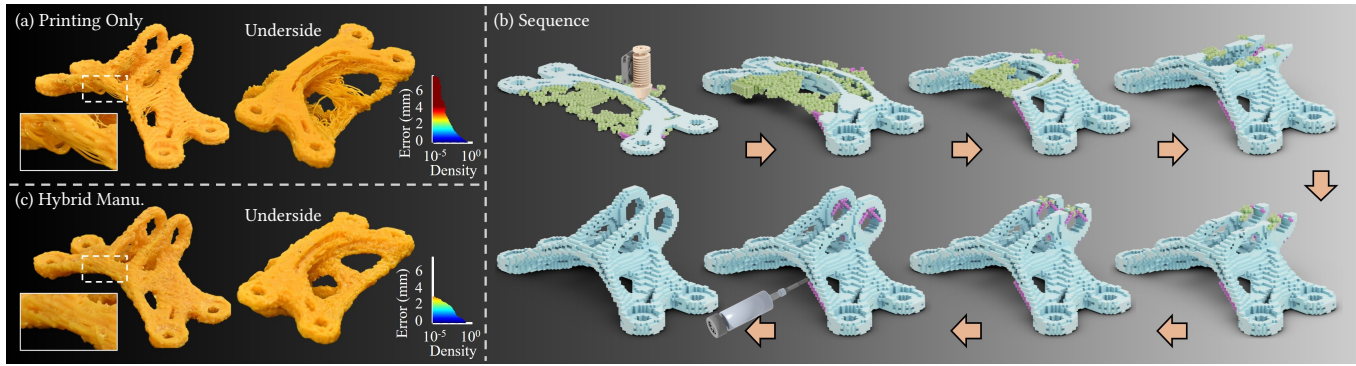
\begin{teaserfigure}
\centering
\includegraphics[width= \textwidth]{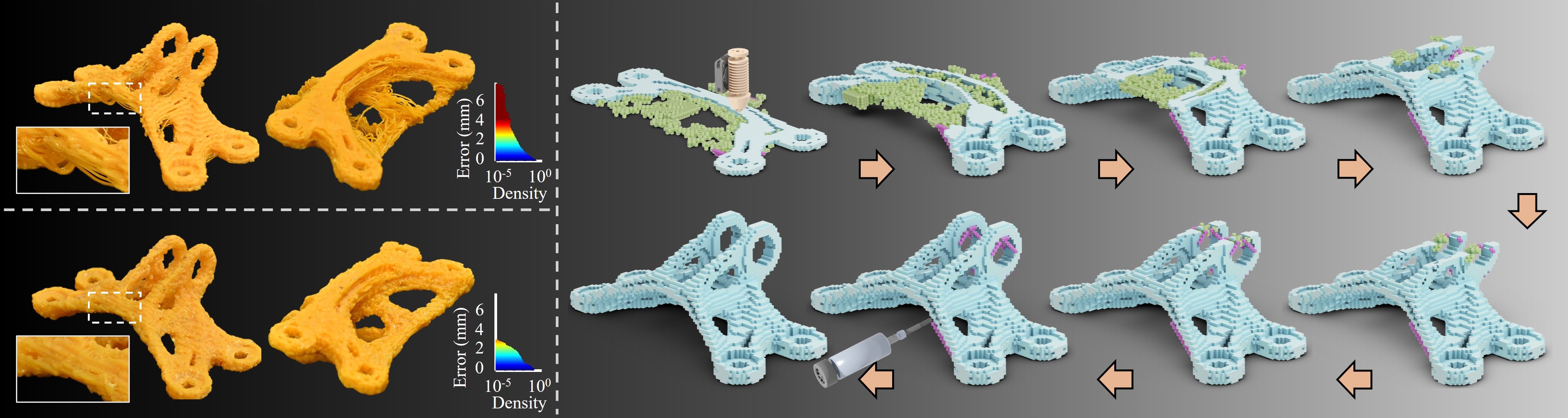}%\vspace{-10pt}
% \put(-507,126){\footnotesize \color{white}(a) Printing Only}
% \put(-507,59){\footnotesize \color{white}(c) Hybrid Manu.}
% \put(-335,126){\footnotesize \color{white}(b) Sequence}
% \put(-375,7){\footnotesize \color{white}Underside}
% \put(-375,72){\footnotesize \color{white}Underside}
\put(-507,128){\footnotesize \color{white}(a) Printing Only}
\put(-507,60){\footnotesize \color{white}(c) Hybrid Manu.}
\put(-326,128){\footnotesize \color{white}(b) Sequence}
\put(-420,54){\footnotesize \color{white}Underside}
\put(-420,124){\footnotesize \color{white}Underside}
\caption{Existing manufacturing approaches often fail to fabricate complex geometries using a single modality, as illustrated by the unsuccessful 3D printing of the GE-Bracket model (a). We propose a process planner for hybrid additive-subtractive manufacturing that enables the successful fabrication of arbitrarily complex shapes. The sequence of additive and subtractive operations determined by our planner is visualized in (b), where temporary supports generated \& removed during the hybrid process are shown as green voxels. The supports pre-added in the preparation stage are highlighted in purple, which can be removed at the last by subtractive operations.
}\label{fig:teaser}
%\vspace{5pt}
\end{teaserfigure}




\maketitle
\section{Introduction}
Additive manufacturing (AM), or 3D printing, offers unmatched flexibility in fabricating complex geometries but still cannot fully avoid support structures for large overhangs, even with multi-axis printing \cite{Dai2018SIG,Zhang2022SIGAsia,Liu2024Sig}. Subtractive manufacturing (SM), such as CNC milling, delivers speed and precision but is limited by tool accessibility \cite{MahdaviAmiri2020,CHEN2023MSE,CHEN2024RCIM}. Recent advances in hybrid manufacturing (HM) systems that combine AM and SM integrate the geometric freedom of AM with the accuracy of SM \cite{RABALO2023}, enabling supports to be printed and then effectively removed, thereby expanding the range of manufacturable shapes.

Despite this promise, current HM planning methods remain limited. Many approaches utilize SM solely for surface finishing, which does not address the challenge of realizing shapes that cannot be 3D printed without support (e.g., \cite{Zhao2023}). Other methods (e.g., \cite{PARC2022}) yield results with inevitable under-cuts or over-cuts, making it impossible to guarantee that the final part exactly matches the input geometry. In this paper, we address the fundamental question of feasibility and completeness in hybrid manufacturing: Can any model be fabricated? A computational framework is presented to guarantee the exact reproduction of arbitrary target geometries by generating a sequence of AM \& SM operations with automatic switching  (see Fig.~\ref{fig:teaser}).

The capability to automatically fabricate arbitrary shapes is particularly important for structural forms obtained from topology optimization (TO), since incorporating manufacturing constraints into TO inevitably sacrifices mechanical performance \cite{Langelaar2016AM}. As illustrated in Fig.~\ref{fig:MBBComp}(b), enforcing a self-support constraint during TO leads to structures with higher strain energy, as evident in the color maps. In contrast, our approach enables the automatic fabrication of TO results without imposing such constraints (Fig.~\ref{fig:MBBComp}(a)). It is worth noting that manually adding supports to the unconstrained TO result in Fig.~\ref{fig:MBBComp}(a) may not resolve the issue, as many of these supports cannot be removed by SM tools after printing due to tool-accessibility limitations \cite{zhong2025deepmill,DING2023CAD}.

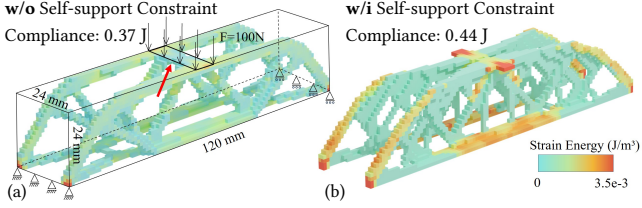
\begin{figure}
    \centering
    \includegraphics[width=\linewidth]{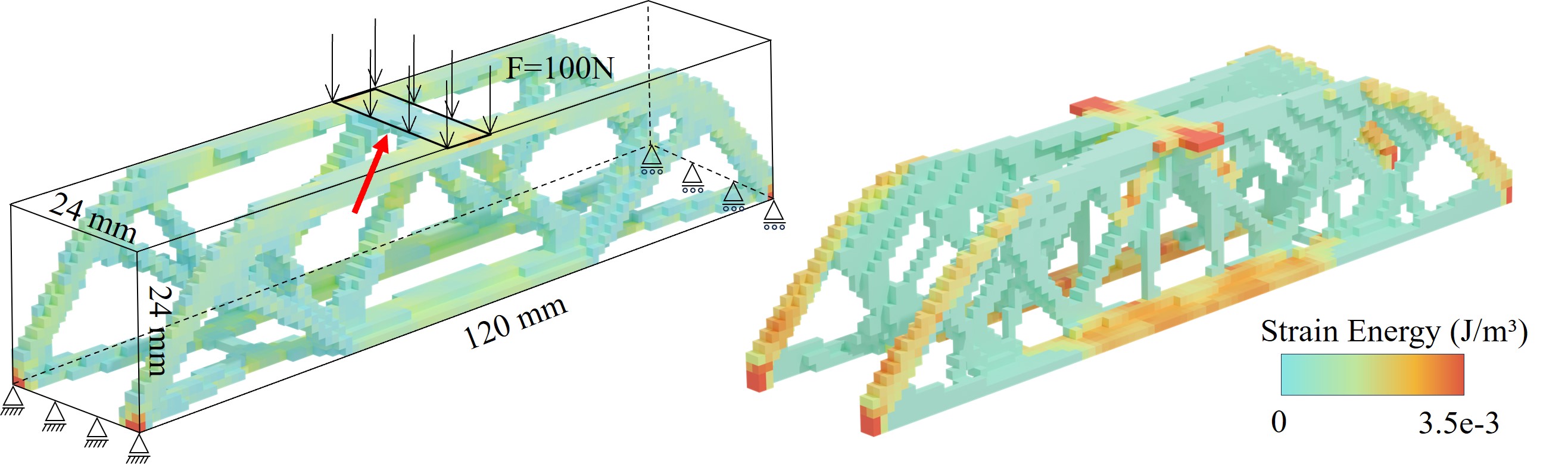}
    \put(-245,1){\footnotesize \color{black}(a)}
    \put(-125,1){\footnotesize \color{black}(b)}
    \put(-246,70){\footnotesize \color{black}\textbf{w/o} Self-support Constraint}
    \put(-246,60){\footnotesize \color{black}Compliance:~$0.37~\mathrm{J}$}
    \put(-117,70){\footnotesize \color{black}\textbf{w/i} Self-support Constraint}
    \put(-117,60){\footnotesize \color{black}Compliance:~$0.44~\mathrm{J}$}
\caption{
Topology optimization of the Messerschmitt-Bölkow-Blohm (MBB) beam (voxel resolution $100 \times 20 \times 20$, volume fraction 10\%) using the method of \cite{Langelaar2016AM}: (a) without self-support constraints and (b) with self-support constraints. Color maps show the strain energy distribution, and the corresponding compliance values used as the TO objective are also reported. The regions highlighted by red arrows in (a) would require AM supports that cannot be removed by SM tools after printing the entire structure -- necessitating a sequence with alternating AM / SM operations.
}\label{fig:MBBComp}
\end{figure}

\subsection{Problem and Challenge}\label{subsec:Challenge}
Given an input solid model $\mathcal{H}$, represented as a set of points in $\mathbb{R}^3$, the planning goal of hybrid manufacturing is to compute a sequence of AM and SM operations that incrementally transform the initial empty model $\mathcal{M}_0 = \emptyset$ into a final model $\mathcal{M}_m$ exactly matching $\mathcal{H}$. This process produces a sequence of intermediate shapes $\mathcal{M}_{t=1,\ldots,m}$, where each $\mathcal{M}_t$ must be self-standing (i.e., without floating or disconnected regions), a condition we refer to as \textit{stable} throughout this paper.
\begin{itemize}
\item When an AM operation is applied to $\mathcal{M}_i$, it produces a new shape $\mathcal{M}_{i+1}$ such that $\mathcal{M}_i \subset \mathcal{M}_{i+1}$. The added material must lie on or above the top layer of $\mathcal{M}_i$ and satisfy self-supporting constraints \cite{Vanek2014CleverSupport}. 

\item When an SM operation is applied to $\mathcal{M}i$, it produces a new shape $\mathcal{M}_{i+1} \subset \mathcal{M}_i$, typically by removing material from the boundary of $\mathcal{M}_i$ and should satisfy the tool accessibility requirement \cite{zhong2025deepmill,Harabin2023Accessibility}. 
\end{itemize} 
The challenge lies in searching for a feasible sequence of AM and SM operations, which collectively form a long planning chain governed by local feasibility and geometric constraints. AM steps are typically guided by the difference between the current shape $\mathcal{M}_i$ and the target model $\mathcal{H}$, while SM steps are employed to remove temporary support structures after they have served their purpose in supporting overhanging regions. However, forward search strategies can easily fall into topological dead-ends -- for instance, when a region that needs to be printed requires support that was not previously added (see Fig.~\ref{fig:whyInverseBetter}(a) for an illustration). This issue poses a major obstacle to the completeness of planning and is a central challenge addressed in our work.

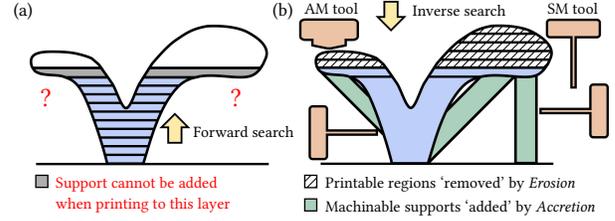
\begin{figure}
\centering
\includegraphics[width=.9\linewidth]{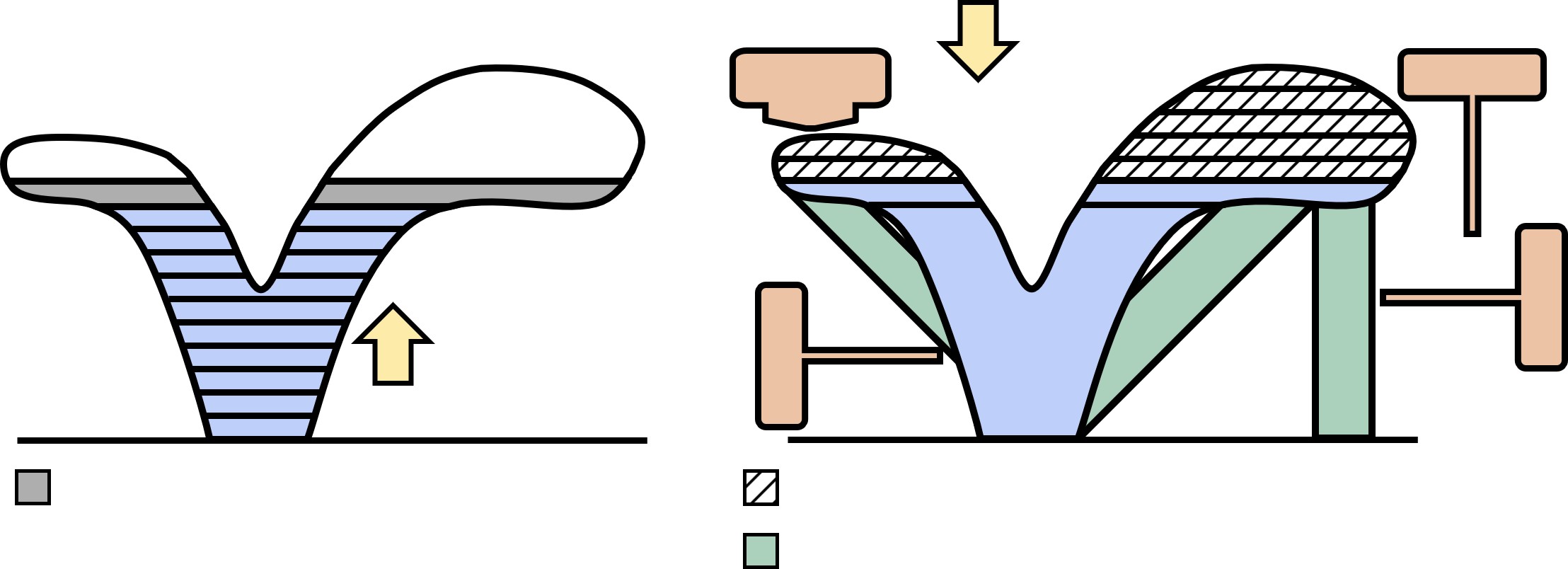}
\put(-225,74){\footnotesize \color{black}(a)}
\put(-127,74){\footnotesize \color{black}(b)}
% \put(-210,65){\scriptsize \color{red}Support cannot be added}
% \put(-210,59){\scriptsize \color{red}when printing to this layer}
\put(-209,9){\scriptsize \color{red}Support cannot be added}
\put(-209,1){\scriptsize \color{red}when printing to this layer}
\put(-143,40){\large \color{red}?}
\put(-215,40){\large \color{red}?}
\put(-157,28){\scriptsize \color{black}Forward search}
\put(-74,74){\scriptsize \color{black}Inverse search}
\put(-107,0){\scriptsize \color{black}Machinable supports `added' by \textit{Accretion}}
%\put(-67,4){\scriptsize \color{black}Operators}
\put(-107,9){\scriptsize \color{black}Printable regions `removed' by \textit{Erosion}}
\put(-115,75){\scriptsize \color{black}AM tool}
\put(-23,75){\scriptsize \color{black}SM tool}
%\put(-45,70){\scriptsize \color{black}}
%\put(-35,69){\scriptsize \color{black}Operators}
\caption{This example illustrates why planning a sequence of inverse AM and SM operations is advantageous: (a) forward search encounters a topological dead-end, making it impossible to complete the full shape; (b) inverse search enables the generation of temporary, removable supporting structures (highlighted in green), allowing exact fabrication of the input model.%\charlie{more discussion is needed -- some problem found (17/05)} \yongxue{updated}
}\label{fig:whyInverseBetter} 
\end{figure}

\subsection{Our Method and Contribution}\label{subsec:OurMethod}
To overcome the topological obstacles that often arise in forward search, we propose a novel strategy that performs the planning process in an inverse direction. Specifically, we define the inverse of a feasible SM operation as an \textit{accretion} operation and the inverse of a feasible AM operation as an \textit{erosion} operation. The searching objective is to compute a nullification sequence -- an ordered list of accretion and erosion operations that progressively reduce the input model to the empty model $\mathcal{M}_0 = \emptyset$. This inverse formulation allows us to use the accretion operation to `add' removable supports beneath regions that would otherwise require support -- i.e., regions with large overhang. As newly added structures may themselves be overhanging, the accretion operations are applied iteratively until the added structure becomes stable. These supports will be `removed' later during the inverse search by erosion operations (see Fig.~\ref{fig:whyInverseBetter}(b) for an illustration). 

By introducing a rule that restricts erosion to the topmost layer and applies accretion only below it, we theoretically prove that any model can be nullified through such a sequence. This ensures the completeness of the corresponding forward sequence of AM and SM operations for fabricating arbitrary geometries. While voxel-based representations have been widely adopted in advanced product design (e.g., \cite{Doubrovski2015CAD,SigmundNature}) and manufacturing process planning (e.g., \cite{Dai2018SIG}), we employ voxels in this work for their simplicity and ease of implementation. This choice supports effective computation and enables practical validation of our approach. Moreover, we develop a scalable algorithm to compute the sequence while aiming to reduce the number of operations, thereby shortening the overall manufacturing time. 

Our work introduces the following technical contributions. 
\begin{itemize}
\item We propose a novel method to compute feasible sequences of hybrid additive-subtractive manufacturing by inverting operations from the target model to null, while ensuring manufacturability at every step.

\item We prove the completeness of our method: any model can be fabricated exactly -- without excess deposition or machining -- through a valid sequence of additive and subtractive operations generated by our algorithm. 

\item We develop a localized stability check algorithm and a heuristic pre-processing strategy that, respectively, improve the scalability of the method and reduce redundant operations. 
\end{itemize}
To the best of our knowledge, this is the first approach that theoretically guarantees the feasibility of hybrid manufacturing for arbitrary shapes. We demonstrate the effectiveness of our method through numerical simulations and physical experiments across a variety of models. The source code of our implementation is available at: \url{https://github.com/Yongxue-Chen/hybManuAccEro}. 

\subsection{Related Work}\label{subsec:RelatedWork}
This section reviews the most relevant research methods in the literature. We focus on prior work that closely relates to our approach, rather than providing an exhaustive survey.

\subsubsection{Voxel-based representation for manufacturing}

Due to its simplicity, voxel-based representation has been widely used in computer graphics for various applications~\cite{Coeurjolly2018Sig}. With the advancement of 3D printing technologies, voxels have also gained popularity in the AM community, particularly for multi-material~\cite{Bader2018sciadv,Doubrovski2015CAD} and multi-axis 3D printing~\cite{Dai2018SIG}. In these applications, voxel grids are used to encode material properties and temporal information for manufacturing sequences. Similarly, voxel-based representations have been adopted in five-axis CNC milling to compute optimized toolpaths and machining sequences (e.g.,~\cite{Tang2020JCISE}), due to their ease of implementation for accessibility analysis and feature size estimation.

More recently, voxels have been applied in hybrid manufacturing research~\cite{PARC2022,Nishiyama2023}, where they facilitate the representation of time-varying model states during the manufacturing process. The approach proposed in this paper also benefits from the simplicity and flexibility of voxel-based modeling to represent intermediate manufacturing states. Different from those researches for toolpath generation (e.g.,~\cite{Zhao2018,Barton2021,Chen2025Tase}), we focus more on the manufacturing sequence planning in this work.

\subsubsection{Process planning for hybrid manufacturing}
Similar to the widely adopted strategy in traditional manufacturing process planning~\cite{Han2000}, feature-based approaches have also been applied to hybrid process planning in earlier research (e.g.,~\cite{HE2023,Liu2020JMS,Parc2018}). However, the completeness of these methods heavily depends on the success of manufacturing feature recognition -- a task that remains challenging even when applying the modern technology of machine intelligence~\cite{ZHANG2024CAGD} or the shape decomposition for manufacturing~\cite{Hu2014Tog}. Furthermore, these approaches are generally limited to models with relatively simple and regular geometries, and are not well-suited for complex freeform surfaces, such as those generated through topology optimization~\cite{Sigmund2001} (see also Fig.~\ref{fig:teaser} for an example). In contrast, our approach is general and does not impose limitations on the geometric complexity of the input models.

In the literature, many existing works focus on using AM operations to construct the basic geometry, followed by SM operations to enhance dimensional accuracy and surface finish. Decomposition-based methods are commonly employed for such hybrid AM-SM strategies, typically aiming to minimize the number of transitions between AM and SM processes. For example, decomposition techniques was proposed in ~\cite{Tang2018CG, Tang2020CAD, Tang2025} for generating collision-free hybrid manufacturing plans, with a particular focus on managing integrated workflows for components with columnar geometries. Zhong et al.~\shortcite{Zhao2023} introduced a decomposition algorithm based on dynamic directed graphs and beam search, which was validated on a 5-axis hybrid manufacturing platform and demonstrated effective sequencing for a variety of complex 3D shapes. While minimizing transitions between AM and SM operations is indeed important, these methods fall short of addressing a more fundamental challenge: how to guarantee the manufacturability of arbitrary geometries using hybrid AM-SM processes. 

To address this, the authors of~\cite{PARC2022} introduced a tree-based planner that incorporates both support-free (for AM operations) and accessibility constraints (for both AM and SM) into the planning loop. However, as discussed earlier, the forward search strategy employed in their approach can easily lead to intermediate shapes with topological dead-ends. Our proposed method overcomes this limitation through an inverse search strategy, enabling the generation of feasible AM-SM process plans even for complex geometries.

\subsubsection{Manufacturing constraints for design} 
Shape and topology optimization has been widely used to address different engineering problems \cite{ShadowArt,Li2023Tog} across various scales~\cite{Zhu2017TwoScale, SigmundNature}. However, the complex geometries generated by these methods often pose significant manufacturing challenges -- particularly for single-material AM -- unless techniques involving soluble support materials in polymer-based AM are employed. To address this issue, researchers have developed methods that incorporate manufacturing constraints directly into the topology optimization process. These include constraints for self-supporting structures in AM~\cite{Langelaar2016AM}, tool accessibility in SM~\cite{Lee2022CMAME}, and multiple constraints for multi-axis hybrid manufacturing~\cite{Liu2024CMAME,Mirzendehdel2022CAD}. The inclusion of such constraints often leads to trade-offs, compromising the performance of the optimized physical properties. In contrast, our proposed hybrid manufacturing approach enables the fabrication of arbitrary geometries derived from topology optimization with isotropic materials -- without sacrificing design performance. 

\section{Voxel-Based Formulation}\label{sec:VoxelFormulation}
This section will introduce the basic formulation as preparation for computing the feasible additive-subtractive manufacturing sequences for models represented as sets of voxels.

\subsection{States and Stability}\label{subsec:StateFunc}
A model to be fabricated is represented as a set of voxels discretized in the computational domain $\Omega$ on a regular voxel grid with resolution $(n_x, n_y, n_z)$. Each voxel is indexed as $v_{i,j,k}$, where $(i, j, k)$ denotes its position in the grid. The current state of the model during fabrication is described by a binary function $M(\cdot)$, referred to as the \textit{state function}, which maps each voxel to its occupancy status. For any voxel $v_{i,j,k}$, the state function is defined as:
\begin{equation}
M(v_{i,j,k})=
\left\{\begin{array}{ll}
    1 & \text{if $v_{i,j,k}$ is solid,} \\
    0 & \text{if $v_{i,j,k}$ is empty.}
\end{array}
\right. 
\end{equation}
We denote the target model state as $M_{\mathcal{H}}$, which exactly matches the voxelized representation of the input model $\mathcal{H}$. Similarly, the \textit{null} state is denoted by $M_{\text{null}}$, where $M_{\text{null}} \equiv 0$ for all voxels. The volume of a model with state $M(\cdot)$ is computed as the total number of solid voxels, i.e., $\sum_{v_{i,j,k}} M(v_{i,j,k})$.

The \textit{stability} of a model under fabrication is defined below using the concept of voxel connectivity.

\begin{definition}
Two solid voxels are considered \textit{connected} if they satisfy the following conditions:
\begin{enumerate}
    \item They are either face-neighbors or edge-neighbors in the voxel grid;
    \item Connectivity is transitive -- i.e., if two voxels are each connected to a common voxel, they are considered connected to each other, 
\end{enumerate}
\end{definition}
\noindent We restrict connectivity to face and edge neighbors because corner-based (vertex-only) contacts provide insufficient material adhesion to ensure structural stability during fabrication. The set of face and edge neighbors of voxel $v_{i,j,k}$ is denoted as $\mathcal{N}(v_{i,j,k})$.

\begin{definition}
A model is said to be \textit{stable} if every solid voxel is connected (in the above sense) to at least one solid voxel located at the base layer.
\end{definition}
\noindent The most straightforward way to evaluate the stability of a model is to use a flood-fill (region-growing) algorithm that checks whether all solid voxels are reachable from the base layer via valid connections. However, this approach is time-consuming. To address this, we develop a conservative yet local method in Sec.~\ref{subsec:LocalStabilityCheck} to enable efficient stability checking.

\subsection{Tools and Manufacturing Operations}\label{subsec:ManuOperation}
We introduce the AM and SM tools below along with the permitted orientations of their operations, which serve as the prerequisites for defining operation feasibility. 

\begin{prerequisite}
\label{prere:AMTool}
The AM tool is restricted to a vertically downward orientation, allowing material deposition on a single voxel without interfering with other voxels at the same height.
\end{prerequisite}

\noindent Given the relatively flat shape of commonly used 3D printing heads, we define the occupancy set of the AM tool as follows.
\begin{definition}\label{def:AMTool}
When depositing material at voxel $v_{i,j,k}$, the occupancy set of the AM tool is defined as
\begin{center}
    $\mathcal{T}_A(v_{i,j,k})=\{ v_{\alpha,\beta,\gamma} \, | \, \gamma > k\}$.
\end{center}
%which represents all voxels located vertically above $v_{i,j,k}$.
\end{definition}

Differently, more flexibility is introduced for SM operations.

\begin{prerequisite}
\label{prere:SMTool}
The SM tool is permitted to operate along five orientations: the negative $z$-axis, and the $\pm x$ and $\pm y$ axes. The tool can remove material from voxels aligned along its tool axis, penetrating into the model up to a specified depth determined by the tool length $\bar{L}$.
\end{prerequisite}

\noindent Similarly, the occupancy set of an SM tool can be defined according to the tool orientation and the tool length.

\begin{definition}\label{def:SMTool}
When removing material at voxel $v_{i,j,k}$, the occupancy set of the SM tool is defined as one of these five sets:
\begin{enumerate}
\item $\mathcal{T}^{-z}_S(v_{i,j,k})=\{ v_{\alpha,\beta,\gamma} \, | \, \gamma \geq k+\bar{L} \} \bigcup \{ v_{i,j,\gamma} \, | \, \gamma > k \}$;
\item $\mathcal{T}^{+x}_S(v_{i,j,k})=\{ v_{\alpha,\beta,\gamma} \, | \, \alpha \leq i-\bar{L} \} \bigcup \{ v_{\alpha,j,k} \, | \, \alpha < i \}$;
\item $\mathcal{T}^{-x}_S(v_{i,j,k})=\{ v_{\alpha,\beta,\gamma} \, | \, \alpha \geq i+\bar{L} \} \bigcup \{ v_{\alpha,j,k} \, | \, \alpha > i \}$;
\item $\mathcal{T}^{+y}_S(v_{i,j,k})=\{ v_{\alpha,\beta,\gamma} \, | \, \beta \leq j-\bar{L} \} \bigcup \{ v_{i,\beta,k} \, | \, \beta < j \}$;
\item $\mathcal{T}^{-y}_S(v_{i,j,k})=\{ v_{\alpha,\beta,\gamma} \, | \, \beta \geq j+\bar{L} \} \bigcup \{ v_{i,\beta,k} \, | \,  \beta > j \}$.
\end{enumerate}
\end{definition}

The orientations of tools and their corresponding occupancy sets have been illustrated in Fig.~\ref{fig:AM_SM_Tools}. 
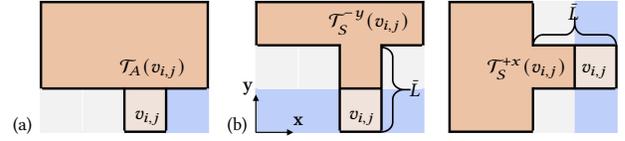
\begin{figure}
\includegraphics[width=0.9\linewidth]{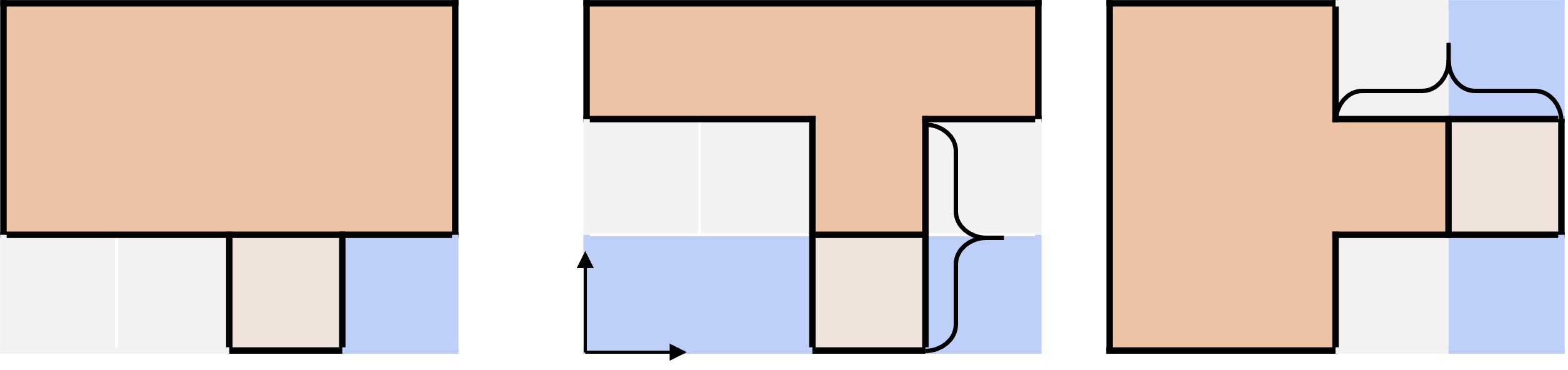}
\put(-229,3){\footnotesize \color{black}(a)}
\put(-148,3){\footnotesize \color{black}(b)}
%\put(-73,3){\footnotesize \color{black}(c)}
%\put(-187,5){\tiny \color{black}$\mathcal{T}_A(v_{i,j})$}
\put(-142,18){\footnotesize \color{black}$\mathbf{y}$}
\put(-123,5){\footnotesize \color{black}$\mathbf{x}$}
\put(-189,25){\footnotesize \color{black}$\mathcal{T}_A(v_{i,j})$}
\put(-183,7){\footnotesize \color{black}$v_{i,j}$}
\put(-110,42){\footnotesize \color{black}$\mathcal{T}^{-y}_S(v_{i,j})$}
\put(-102,7){\footnotesize \color{black}$v_{i,j}$}
\put(-50,25){\footnotesize \color{black}$\mathcal{T}^{+x}_S(v_{i,j})$}
\put(-14,25){\footnotesize \color{black}$v_{i,j}$}
\put(-19,45){\footnotesize \color{black}$\bar{L}$}
\put(-79,16){\footnotesize \color{black}$\bar{L}$}
\caption{2D illustrations of (a) an AM tool and (b) SM tools, shown with tool length $\bar{L}$ represented in voxel form.}
\label{fig:AM_SM_Tools}
\end{figure}
Note that the occupancy sets (i.e., tool shapes) defined above are intentionally conservative. However, the formulation and the associated algorithm are general and remain applicable to any tool shapes that are smaller or more compact than those assumed in this paper.

We now define the feasibility of AM and SM operations, where an AM operation converts an empty voxel into a solid one, while an SM operation does the inverse -- changing a solid voxel into an empty one.

\begin{definition}
An AM operation applied at voxel $v_{i,j,k}$ is \textit{feasible} if it is i) \textit{self-supported} -- i.e., any of its solid neighbors in $\mathcal{N}(v_{i,j,k})$ is located at the height $k-1$, and ii) \textit{collision-free} as $\{ v_{\alpha,\beta,\gamma} \, | \, M(v_{\alpha,\beta,\gamma}) = 1 \} \cap \mathcal{T}_A(v_{i,j,k}) = \emptyset$.
\end{definition}

\noindent AM operations applied to the first layer (i.e., $k=1$) are always self-supported. The self-supporting angle is usually determined by material properties, printing method, and other factors. In all our formulations, the supporting set is defined based on the supporting angle as $\pi/4$. 

\begin{definition}
An SM operation applied at voxel $v_{i,j,k}$ is \textit{feasible} if it is i) \textit{collision-free} as $\{ v_{\alpha,\beta,\gamma} \, | \, M(v_{\alpha,\beta,\gamma}) = 1\}  \cap \mathcal{T}_S(v_{i,j,k}) = \emptyset$, and ii) the resultant model is \textit{stable}.
\end{definition} 
\noindent Here $\mathcal{T}_S(v_{i,j,k})$ can be any occupancy set given in Def.~\ref{def:SMTool}.

\subsection{Problem of Process Planning}\label{subsec:ProbDef}
The problem of hybrid manufacturing process planning is to compute an ordered sequence of interleaved feasible AM and SM operations, using a voxel-based representation, such that the state function of the workspace transitions from the null state $M_{\text{null}}$ to the target state $M_{\mathcal{H}}$, where $M_{\mathcal{H}}$ is defined by the voxelization of the input model $\mathcal{H}$.
\section{Nullification for Process Planning}\label{sec:Nullification}
This section first introduces the operators that serve as the inverses of the AM and SM operations, followed by presenting the nullification algorithm for generating the manufacturing sequence. We then prove the completeness of our nullification algorithm.

\subsection{Accretion and Erosion Operators}\label{subsec:AccEroOperators}
Two inverse operators are introduced to support the nullification process.
\begin{itemize}
\item \textit{Erosion}: converts a solid voxel $v_{i,j,k}$ to empty, serving as the inverse of an AM operation;
\item \textit{Accretion}: converts an empty voxel $v_{i,j,k}$ to solid, serving as the inverse of an SM operation.
\end{itemize}
Each operator transitions the state function from $M_t(\cdot)$ to $M_{t-1}(\cdot)$. Their feasibility is defined as follows.

\begin{definition}\label{def:Ero}
An \textit{erosion} operator is feasible if the following conditions are met: 
\begin{enumerate}
\item The original AM operation at $v_{i,j,k}$ is feasible -- i.e., it satisfies both the \textit{self-support} and \textit{collision-free} conditions;
\item The resultant model represented by $M_{t-1}(\cdot)$ is \textit{stable}.
\end{enumerate}    
\end{definition}

\begin{definition}\label{def:Acc}
An \textit{accretion} operator is feasible when the following conditions are met:
\begin{enumerate}
\item The corresponding SM operation at $v_{i,j,k}$ is \textit{collision-free};
\item There exists at least one solid voxel in the neighborhood $\mathcal{N}(v_{i,j,k})$.
\end{enumerate}
\end{definition}

Without loss of generality, we assume the input model is \textit{stable}. Under this assumption, the above feasibility definitions ensure that each intermediate state $M_t(\cdot)$ during the nullification process also remains stable. Stability does not need to be explicitly verified for accretion operations because the presence of a neighboring solid voxel and the stability of $M_i(\cdot)$ together imply the stability of the resulting $M_{t-1}(\cdot)$. As discussed, stability checking is computationally expensive and should be avoided whenever possible for algorithmic scalability.

\subsection{Sequence Generation}\label{subsec:SequenceGen}
A feasible hybrid manufacturing sequence for the input model $M_{\mathcal{H}}$ is generated by progressively applying erosion and accretion operators, transforming $M_{\mathcal{H}}$ step-by-step into the null model $M_{\text{null}}$ (as illustrated in Fig.~\ref{fig:NullificationProcess}).

\begin{figure}
\centering
\includegraphics[width=\linewidth]{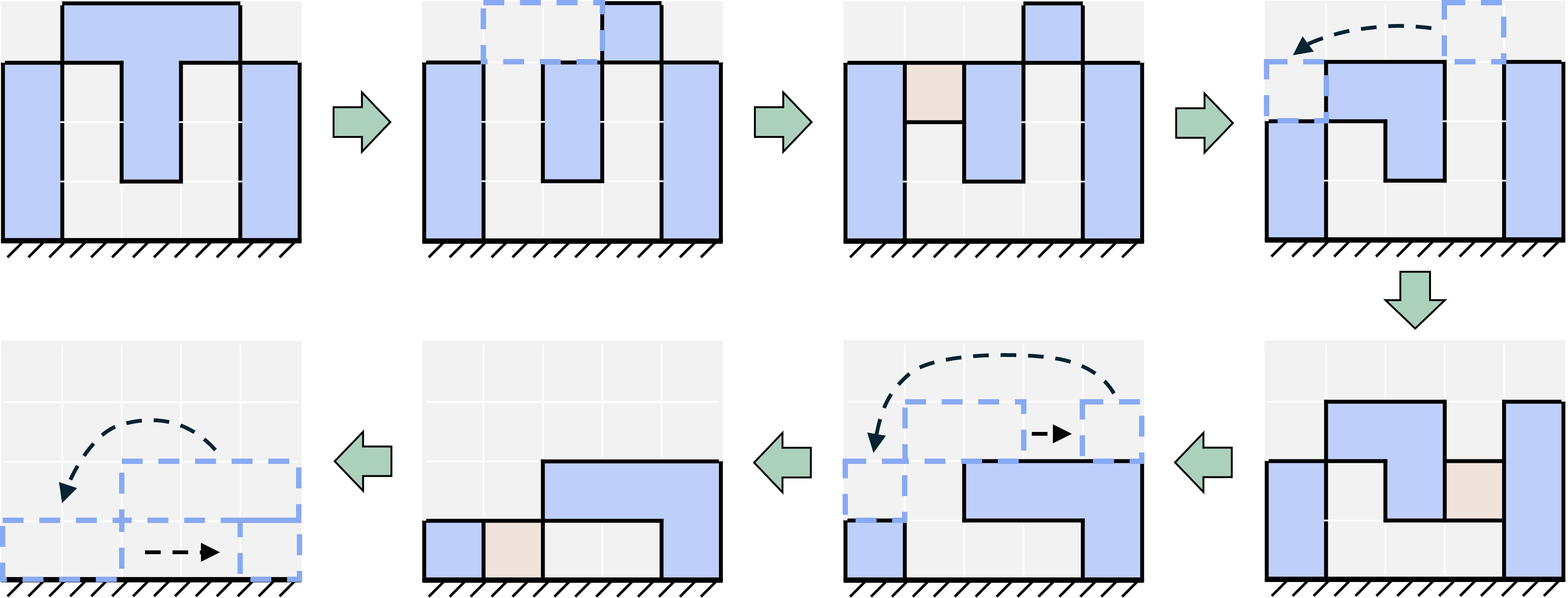}
\put(-192,28){\footnotesize \color{black}$Ero$}
\put(-128,28){\footnotesize \color{black}$Acc$}
\put(-62,28){\footnotesize \color{black}$Ero$}
\put(-15,44){\footnotesize \color{black}$Acc$}
\put(-192,80){\footnotesize \color{black}$Ero$}
\put(-128,80){\footnotesize \color{black}$Acc$}
\put(-62,80){\footnotesize \color{black}$Ero$}
\caption{Illustration of the nullification process in the 2D case, where the dashed blue boxes represent the results of applying erosion operations (denoted by $Ero$) with black dashed arrows giving the sequence of operations. Voxels `added' by the accretion operations (denoted by $Acc$) are displayed in light beige.
}\label{fig:NullificationProcess}
\end{figure}

\begin{figure}
\centering
\includegraphics[width=1.0\linewidth]{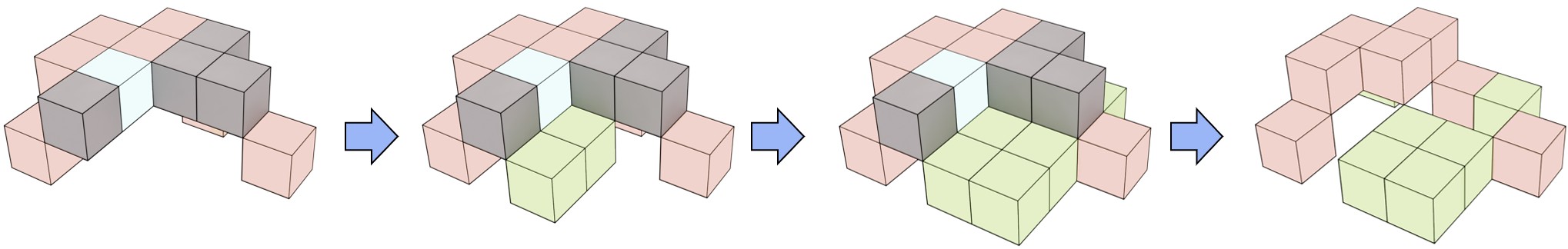}
\put(-224,12){\footnotesize \color{black}$\uparrow$}
\put(-225,6){\footnotesize \color{black}$v_{i,j,K}$}
\put(-192,25){\footnotesize \color{black}$Acc$}
\put(-128,25){\footnotesize \color{black}$Acc$}
\put(-62,25){\footnotesize \color{black}$Ero$}
\caption{For enhancing the erosion-feasibility of a voxel $v_{i,j,K}$ displaye in light blue color -- i.e., $v_{i\pm 1,j,K-1}$ and $v_{i,j\pm 1,K-1}$ are all empty, the accretion operators are applied ring by ring to the neighbors of $\mathcal{N}(v_{i,j,K})$ so that the voxel $v_{i,j,K}$ becomes erosion-feasible at the end of `growing'. The voxels `added' by accretion are displayed in green color, and the voxels in the erosion feasible set $\Lambda$ are displayed in dark grey. Note that all voxels in $\Lambda$ are considered as empty when checking erosion-feasibility; therefore, the second batch of accretion operations are needed in this case. After applying the feasible accretion operations, all voxels in $\Lambda$ are removed together by the erosion operations. %\tao{updated 17/08}
}\label{fig:ringByringAccretion}
\end{figure}

\begin{figure*}
\centering
\includegraphics[width=1\linewidth]{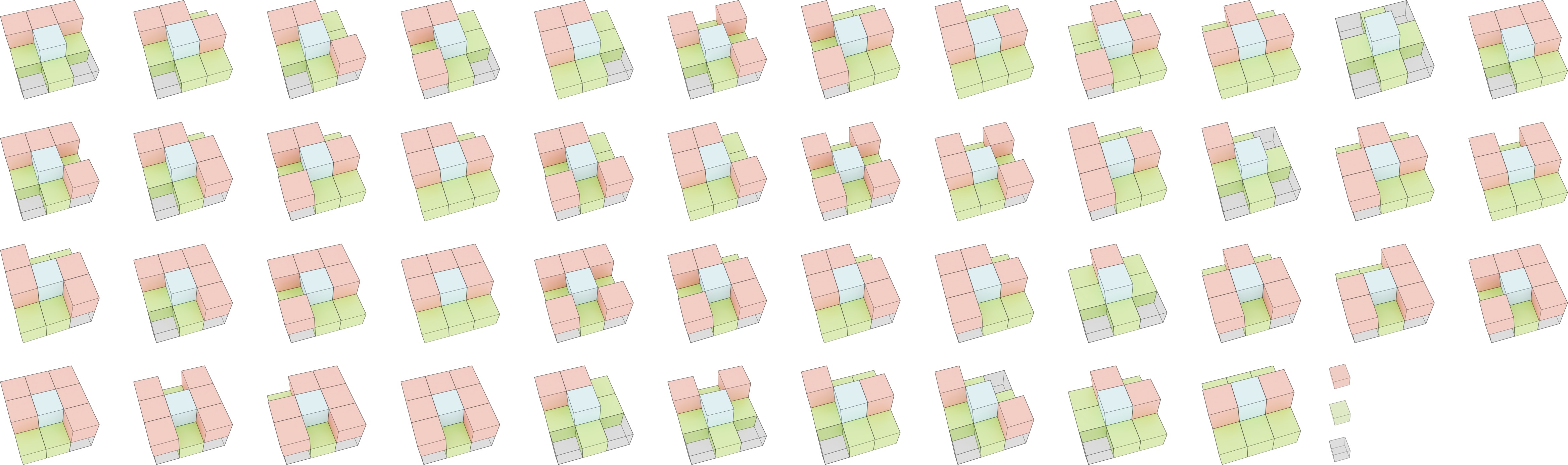}
\put(-69,27){\footnotesize \color{black} Existing solid neighbors}
\put(-69,15){\footnotesize \color{black} Newly `added' voxels}
\put(-69,3){\footnotesize \color{black} Non-contributing voxels}
\caption{All 46 configurations -- excluding symmetric cases -- of existing voxels (peach) around a boundary voxel $v_{i,j,k}$ (blue), and the voxels in green are added by feasible-accretion operations to make $v_{i,j,k}$ erosion-feasible. Voxels shown in gray are connected to $v_{i,j,k}$ only by a vertex and are not considered as neighbors providing the support to $v_{i,j,k}$; therefore, their solid or empty status does not affect the erosion-feasibility of $v_{i,j,k}$.
}\label{fig:completenessProof}
\end{figure*}

The order of operators is determined heuristically based on the following consideration:
\begin{itemize}
\item The erosion operator is only applied on the top layer of the current model $M_t(\cdot)$ because of the relatively flat shape of the printer head commonly used in AM operations;

\item The accretion operations will first be conducted to enable more voxels at the top layer to be `removed' by erosion; 

\item Whenever possible multiple erosion operations are applied consecutively to preserve the continuity of AM operations on the top layer.  
\end{itemize}
The following algorithm with 5 major steps is developed by considering the above heuristics.
\begin{itemize}
\item \textbf{Step 1:} Initialize an empty ordered list of operators, $\Gamma$.

\item \textbf{Step 2:} Initialize an empty set $\Lambda$ to store erosion-feasible voxels located on the topmost layer of the current model $M_t$ (i.e., the voxels with their height-indices as a constant $K$).

\item \textbf{Step 3:} For each solid voxel $v_{i,j,K}$, check whether it satisfies the feasibility conditions for erosion. If so, add $v_{i,j,K}$ to the set $\Lambda$.

\item \textbf{Step 4:} For each solid voxel $v_{i,j,K}$ that is not erosion-feasible (e.g., the case shown Fig.~\ref{fig:ringByringAccretion}), attempt to improve its feasibility through stability-enhancing accretion operations:
\begin{enumerate}
    \item Construct a voxel set $\mathcal{A}$, initially containing $v_{i,j,K}$ and all solid voxels in its neighborhood $\mathcal{N}(v_{i,j,K})$;

    \item Iteratively apply feasible accretion operations -- checking all five orientations -- to neighboring voxels of $\mathcal{A}$ that lie below level $K$, adding the resulting solid voxels to $\mathcal{A}$; %\yongxue{check using all five orientations}

    \item Repeat this process ring-by-ring until $v_{i,j,K}$ becomes erosion-feasible by considering all voxels in $\Lambda$ as empty when checking erosion-feasibility\footnote{This is a conservative consideration as we cannot determine the order of following erosion operators in advance.}. Then, append all the applied accretion operations to $\Gamma$, update the state function of the model accordingly, and add $v_{i,j,K}$ to $\Lambda$;

    \item If no further feasible accretion operations can be applied and $v_{i,j,K}$ remains \textit{not} erosion-feasible, the algorithm will move to enhance the next voxel's erosion-feasible.
\end{enumerate}

\item \textbf{Step 5:} Sort all voxels in $\Lambda$, and sequentially apply erosion operations to each (see the right of Fig.~\ref{fig:ringByringAccretion}). Append these operations to $\Gamma$ and update the model state.

\item \textbf{Step 6:} Repeat from Step 2 until the model reaches the null state $M_{\text{null}}$.
\end{itemize}
At the end of this nullification algorithm, the ordered list $\Gamma$ gives the inverse order of AM and SM operations that can fabricate the target model $\mathcal{H}$ by hybrid manufacturing. In this algorithm, Step 4 plays a critical role in improving erosion-feasibility by introducing accretion operations -- i.e., temporary support structures that are later removed -- thus enhancing printability. An illustration can be found in Fig.~\ref{fig:ringByringAccretion}. In Step 5, erosion operations are preferably applied to adjacent voxels to preserve continuity in AM operations. 

\subsection{Completeness of Nullification}\label{subsec:NullifierCompleteness}
We now prove the \textit{completeness} of the nullification algorithm introduced above. The completeness is primarily guaranteed by two key properties of the algorithm:
\begin{itemize}
\item The accretion operations are restricted to voxels located below the topmost layer of the current model;

\item Any voxel located at the boundary of the topmost layer can always be converted into erosion-feasible by accretion operations using an SM tool with length $\bar{L} \geq 2$.
\end{itemize}
The first property stems directly from the algorithm’s design, which prevents accretion from affecting the current top layer. The second property can be proven by analyzing the conditions at a boundary voxel on the topmost layer -- details of this proof are provided below with the help of Fig.~\ref{fig:completenessProof}.

Without loss of generality, a boundary solid voxel $v_{i,j,K}$ at the topmost layer of the current state satisfies one or more of the following conditions: $M_t(v_{i \pm 1,j,K})=0$, $M_t(v_{i ,j \pm 1,K})=0$. According to the algorithm introduced in Sec.~\ref{subsec:SequenceGen}, it is possible to add new voxels at $v_{i \pm 1, j \pm 1, K-1}$ (i.e., the green voxels shown in Fig.~\ref{fig:completenessProof}) from the top, as the tool length $\bar{L} \geq 2$. These accreted voxels will form a connected structure surrounding $v_{i,j,K}$ (including both the green and peach voxels in Fig.~\ref{fig:completenessProof}), ensuring that the resultant model remains stable after applying the erosion operator to $v_{i,j,K}$. Furthermore, since $v_{i,j,K}$ is on the boundary, at least one of the voxels $v_{i \pm 1,j,K-1}$ or $v_{i,j \pm 1,K-1}$ will be added, providing support from below. This guarantees that $v_{i,j,K}$ becomes self-supported. Consequently, the erosion of $v_{i,j,K}$ becomes feasible.

\begin{wrapfigure}[9]{r}{0.13\linewidth}
% \begin{center}
\vspace{-\intextsep}
\centering
\hspace{-18pt}
% \vspace{-20pt}
\includegraphics[width=\linewidth]{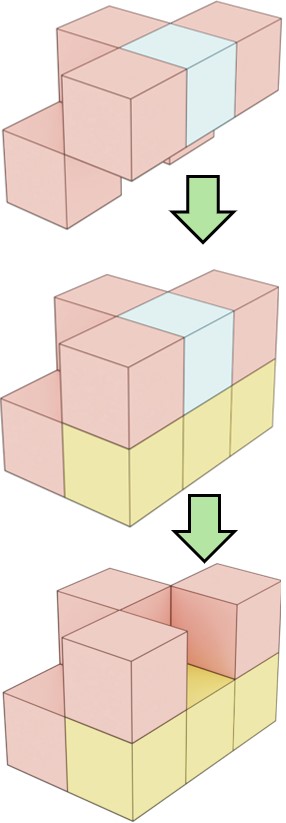}
\put(-5,69){\scriptsize \color{black}$Acc$}
\put(-5,33){\scriptsize \color{black}$Ero$}
% \end{center}
\end{wrapfigure}
Note that for voxels located at the boundary of the working space $\Omega$, accretion operations are usually performed via SM operations along the horizontal directions. The case of accretion by using a horizontal SM tool has been illustrated in the insert figure, where the accretion-added support voxels to be `removed' by a horizontal SM tool are displayed in yellow color. After these accretion operations, the blue voxel becomes erosion-feasible and can be removed. Moreover, when only one layer remains in the model under nullification, all voxels are erosion-feasible as they are supported by the working base.

\section{Scalability and Redundancy}\label{sec:Heuristics}
This section introduces two extended schemes designed to enhance the scalability of our approach: a localized stability check to reduce computational cost, and a pre-processing step to heuristically reduce redundancy in temporary supporting structures.

\subsection{Localized Stability Check}\label{subsec:LocalStabilityCheck}
In the nullification algorithm presented in Sec.~\ref{subsec:SequenceGen}, the most time-consuming step is the stability check. To enhance the scalability of the algorithm for high-resolution voxel models, we introduce two key modifications. First, instead of checking stability after each accretion operation, we perform the stability check after taking a sequence of feasibility enhancing operations. As accretion operations in Step 4 of the nullification algorithm proceed ring by ring, we adopt a ring-wise stability checking strategy. Second, to avoid repeatedly performing global operations for region-growing from the base -- which are computationally expensive -- we develop a local stability check scheme. Our local check approach ensures that all potentially unstable configurations are detected without requiring a full model traversal, thereby significantly improving efficiency while maintaining correctness in manufacturability. Details are presented below.

The local stability check is developed based on two Lemmas.
\begin{lemma}\label{lemma:RemainStable}
If a solid voxel $v_{i,j,k}$ is removed from a stable model, the stability remains if and only if all of $v$'s neighboring solid voxels are connected to the base in the new model.
\end{lemma}
\begin{proof}
The necessity of this Lemma is determined by the definition of model stability. According to the definition of connectivity, the change in the state of $v$ can only directly affect the connectivity of its neighboring voxels to the base. Therefore, the sufficiency is intuitive.
\end{proof}
\begin{lemma}
\label{lemma:RemainStableS}
If a solid voxel $v_{i,j,k}$ not adjacent to the base is removed from a stable model, the new model remains stable when all neighboring solid voxels of $v_{i,j,k}$ are connected. 
\end{lemma}
\begin{proof}
The stability of the original model ensures $v_{i,j,k}$ is connected to the base. This connectivity must be achieved through its neighboring solid voxels since $v_{i,j,k}$ is not at the base. This means that after removing $v_{i,j,k}$, at least one neighboring solid voxel of $v_{i,j,k}$ remains its connectivity to the base. Thus, if all of $v_{i,j,k}$'s neighboring solid voxels are connected, they are all connected to the base. According to Lemma \ref{lemma:RemainStable}, the new model is stable.
\end{proof}

\begin{definition}\label{def:DeltaNeighbor}
The $\Delta$-neighborhood of a solid voxel $v_{i,j,k}$ are those solid voxels $v_{r,s,t}$ connecting to $v_{i,j,k}$ and satisfying $\| (i-r,j-s,k-t) \|_{\infty} \leq \Delta$.
\end{definition}

We now present the localized stability check algorithm based on the above lemmas and definitions. The algorithm consists of four steps:

\begin{itemize} \item \textbf{Step 1}: Construct a voxel set $\Psi$ by collecting all voxels within the $\Delta$-neighborhood of $v_{i,j,k}$.

\item \textbf{Step 2}: Apply a flooding algorithm within $\Psi$ to identify connected components.

\item \textbf{Step 3}: If there is only one connected component in $\Psi$ and $v_{i,j,k}$ is not adjacent to the base (i.e., $k > 1$), return \textbf{True} to indicate the model remains stable.

\item \textbf{Step 4}: If any connected component in $\Psi$ is not adjacent\footnote{This adjacency detection is performed locally, which may produce conservative results with some misclassifications regarding stability.} to the base, return \textbf{False}; otherwise, return \textbf{True}. 
\end{itemize}
Given that the localized stability check is performed as part of the erosion-feasibility assessment, the algorithm presented in Sec.~\ref{subsec:SequenceGen} needs only to confine its accretion operations (i.e., Step 4-(2)) to the local region within the $\Delta$-neighbors as well.
%\yongxue{Given that the localized stability check is performed as part of the erosion-feasibility assessment, the algorithm presented in Section 3.2 needs only confine its accretion operations (Step 4(2)) to the $\Delta$-neighborhood, obviating the need for operations in the global space}
The accuracy of this localized check inherently depends on the size of the $\Delta$-neighborhood. A larger $\Delta$ leads to a more accurate classification but also takes a higher computational cost. 

An advantage of this algorithm is that: it may mistakenly classify a stable configuration as unstable, but it will never misclassify an unstable configuration as stable. Such conservative misclassifications may lead to suboptimal results in the planning algorithm, but they do not compromise manufacturability. For example, as shown in Fig.~\ref{fig:localStabilityCheck}, when a small neighborhood size such as $\Delta = 1$ is used, the configurations in Fig.~\ref{fig:localStabilityCheck}(b) and (c) may be misclassified as unstable. However, the truly unstable configuration in Fig.~\ref{fig:localStabilityCheck}(d) is correctly detected. 

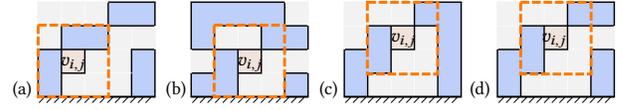
\begin{figure}\centering
\includegraphics[width=.9\linewidth]{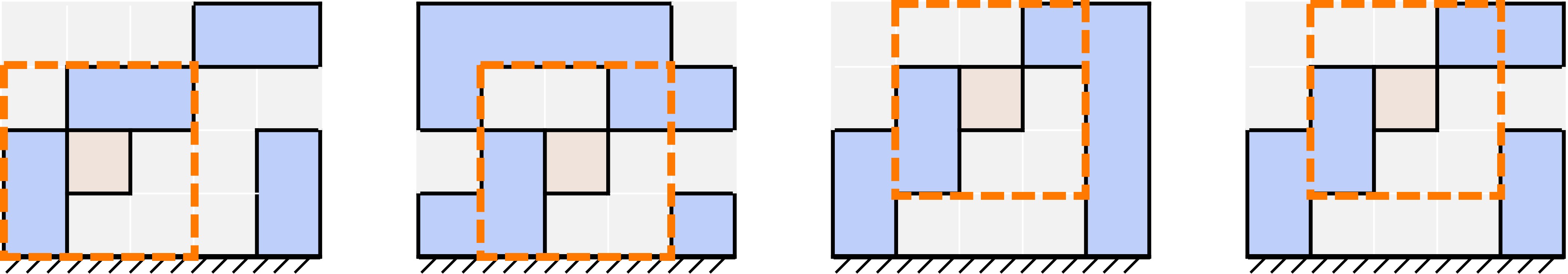}
\put(-228,3){\footnotesize \color{black}(a)}
\put(-170,3){\footnotesize \color{black}(b)}
\put(-112,3){\footnotesize \color{black}(c)}
\put(-55,3){\footnotesize \color{black}(d)}
\put(-210,14){\footnotesize \color{black}$v_{i,j}$}
\put(-143,14){\footnotesize \color{black}$v_{i,j}$}
\put(-85,23){\footnotesize \color{black}$v_{i,j}$}
\put(-28,23){\footnotesize \color{black}$v_{i,j}$}
\caption{Four 2D example cases illustrating the stability of a model after removing a voxel $v_{i,j}$: the cases (a–c) result in a stable model, while (d) leads to an unstable one. The squares with dashed lines indicate the search range for the localized stability check with $\Delta = 1$. When $\Delta = 1$ is used, (b) \& (c) will be misclassified as unstable; but accurate classification can be obtained for (b) \& (c) by using $\Delta = 2$.
}\label{fig:localStabilityCheck}
\end{figure}

\subsection{Pre-Processing: Removal Support}\label{subsec:PreProcessing}
The nullification algorithm presented in Sec.~\ref{sec:Nullification} guarantees the completeness of generating a feasible sequence of AM and SM operations for an input model of arbitrary shape. However, in practical applications, we observe that the algorithm occasionally produces redundant temporary supporting structures -- even though these supports are always eventually removed via SM operations. This redundancy primarily arises from the local search nature of the feasibility enhancing accretion step, which adds support below voxels whose erosion requires additional structural reinforcement. 

To address this issue of redundancy, we introduce a pre-processing step that mitigates the support requirement in a less localized manner. Specifically, regions of the input model $\mathcal{H}$ that require support are first identified. Then, support structures that can be subsequently removed by SM operations are heuristically added through feasible accretion operations, resulting in an enriched model $\mathcal{H}^*$ such that $\mathcal{H} \subset \mathcal{H}^*$. These accretion operations are prepended to the inverse operation sequence $\Gamma$, and the nullification algorithm introduced above is applied to $\mathcal{H}^*$ to further update $\Gamma$. 

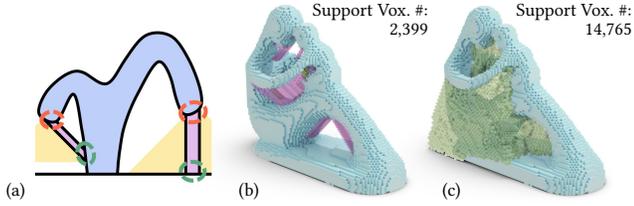
\begin{figure}\centering
\includegraphics[width=.9\linewidth]{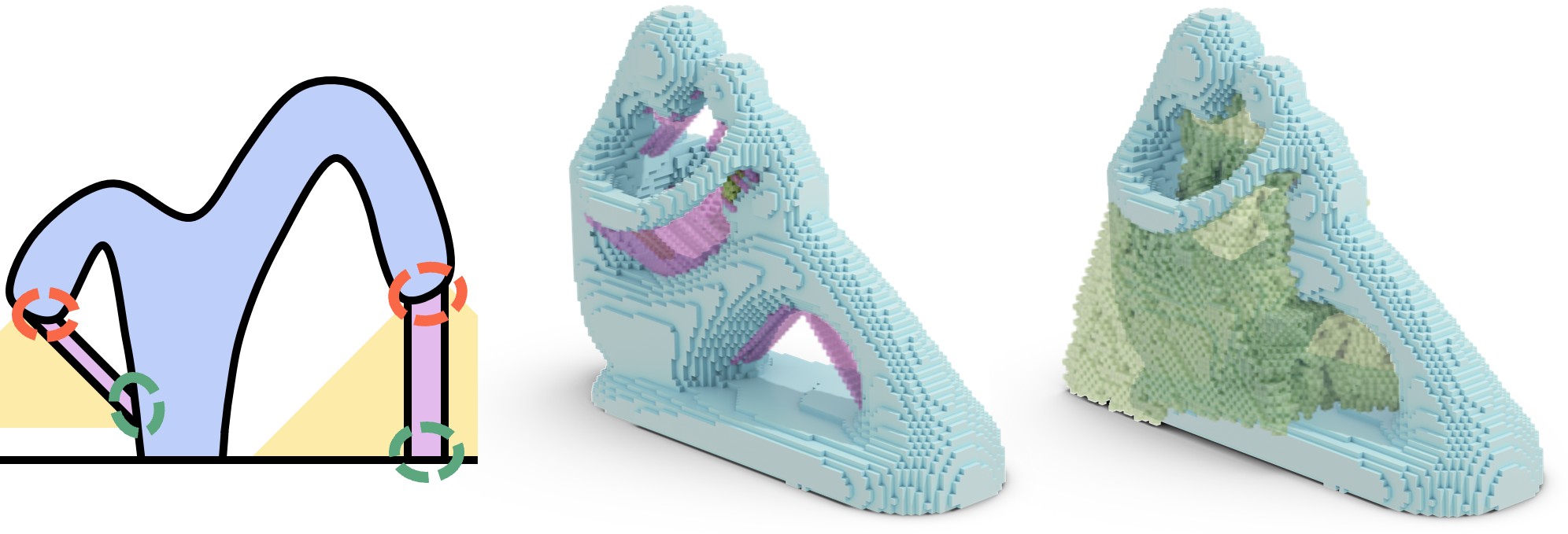}
\put(-230,2){\footnotesize \color{black}(a)}
\put(-142,2){\footnotesize \color{black}(b)}
\put(-65,2){\footnotesize \color{black}(c)}
\put(-37,70){\footnotesize \color{black}Support Vox. \#:}
\put(-11,63){\footnotesize \color{black}14,765}
\put(-114,70){\footnotesize \color{black}Support Vox. \#:}
\put(-85,63){\footnotesize \color{black}2,399}
\caption{Pre-processing to add temporary supports. (a) As illustrated in a 2D example, the regions outlined with red dashed lines exhibit significant overhangs, where removable temporary structures (shown in purple) are generated from the nearest touching points (marked by green dashed circles) within the search region (highlighted in yellow). When applying to the Fertility model represented by a voxel res.: $100 \times 38 \times 74$, the pre-processing step can effectively reduce the number of temporary support voxels by 83.75\% -- (b) with vs. (c) without pre-processing. Note that, this test employs an SM tool with length $\bar{L}=30$.%\charlie{fig to update}\yongxue{updated}
}\label{fig:Preprocessing}
\end{figure}

We present the details of our pre-processing method below with the help of an illustration shown in Fig.~\ref{fig:Preprocessing}(a).
\begin{itemize}
\item \textbf{Step 1:} Identify all voxels in the input model $\mathcal{H}$ that require additional support during AM operations, and collect them into a set $\Phi$. These voxels are then sorted by their distance to the bounding box of the model, prioritizing inner voxels. This is a heuristic rule to reduce the chance that the supports generated for the outer voxels block those of the interior voxels.

\item \textbf{Step 2:} For each voxel $v \in \Phi$, a progressive voxel growth is performed using feasible accretion operations to construct candidate support structures. This growth occurs ring by ring within a downward conical search region below $v$, defined by a cone angle determined by the self-supporting threshold of AM (assumed to be $45^\circ$ under voxel representation). The process continues until the growing structure $\mathcal{S}(v)$ reaches either the ground or a self-supported region of $\mathcal{H}$ -- the first contact location is referred to as the \textit{touching point} (see Fig.~\ref{fig:Preprocessing}(a)). The voxels along the shortest path from the touching point to $v$ are then converted into solids, and the corresponding accretion operations are appended to the beginning of the inverse operation sequence $\Gamma$. 

\item \textbf{Step 3:} If no valid accretion path can be found for voxel $v$, meaning it cannot be supported in a way that allows removal via SM at the end of fabrication, the process skips $v$ and continues with the next voxel in $\Phi$.

\item \textbf{Step 4:} Repeat Steps 2 and 3 until all voxels in $\Phi$ have been processed.
\end{itemize}
Although this is a heuristic algorithm inspired by how humans intuitively perform AM and SM operations to fabricate geometrically complex models, we find that it can effectively reduce the total number of operations -- see the example as shown in Fig.~\ref{fig:Preprocessing}(b)-(c).
\section{Implementation Details}\label{sec:Implementation}

\subsection{Post-Processing: Toolpath Generation}\label{subsec:PostProc}
After generating the inverse sequence of AM and SM operations using the nullification algorithm with voxels as the unit of operation, we convert this sequence into actual AM and SM toolpaths during the post-processing step. In practice, the physical dimensions of AM and SM tools may differ from the voxel size. For example, in all our test cases, voxels are defined with a unit size of $1.2\text{mm}$, which matches the diameter of the SM tool ($1.2\text{mm}$) but is twice the diameter of the AM nozzle ($0.6\text{mm}$).

To account for this difference, neighboring voxels corresponding to AM operations are grouped into AM `patches', and continuous zig-zag AM toolpaths are generated patch by patch. A similar strategy is applied to SM operations by grouping their corresponding voxels into SM `patches'. However, since the orientation of the SM tool can vary across operations, the grouping process must consider both spatial and orientation continuity -- i.e., SM operations with different orientations must be assigned to separate patches.

\subsection{Extension for Filament-Based AM}\label{subsec:AMPrefer}
When applying feasible erosion operations to generate the voxel set that determines the sequence of AM operations (i.e., Steps 3–5 of the nullification algorithm), the minimum printable feature size (MPFS) required by filament-based AM processes -- such as fused deposition modeling -- is not explicitly enforced. While our hybrid manufacturing planning algorithm is general, it can be adapted to incorporate such constraints. For example, erosion operations can be restricted to cases where at least $M$ connected voxels on the top layer can be `removed' together when it is possible, ensuring compatibility with the minimum feature size required for AM. 

\begin{wrapfigure}[5]{r}{0.55\linewidth}
\vspace{-\intextsep}
\centering
\hspace{-18pt}
\includegraphics[width=1.1\linewidth]{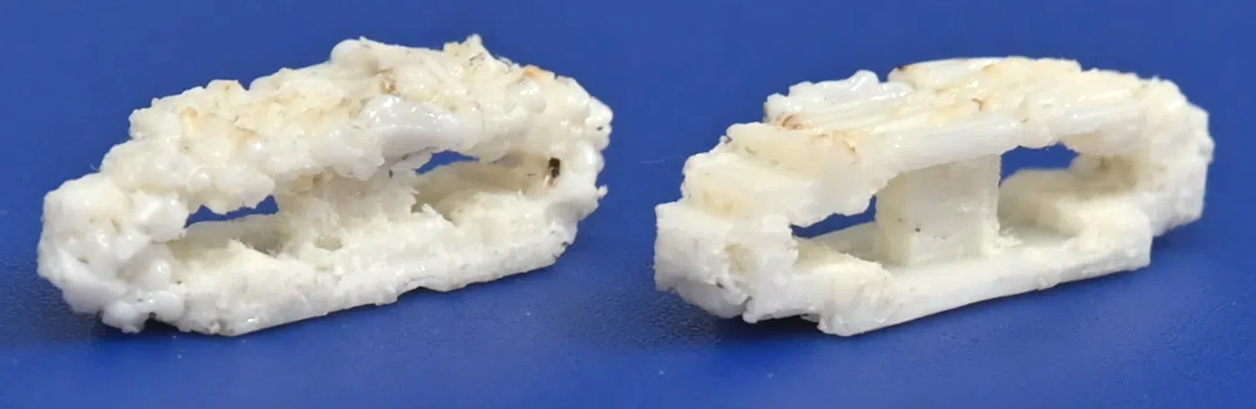}
\put(-130,2){\scriptsize \color{white}MPFS Not Considered}
\put(-49,2){\scriptsize \color{white}MPFS Considered}
\end{wrapfigure}
To enhance the connectivity of voxels that are feasible for erosion, we apply extra accretion operations as a `pre-processing' step before we start to nullify a new layer. It is important to note that this heuristic is optional. Scattered voxels may still need to be eroded individually to maintain the completeness of our nullification algorithm; however, such a case happens rarely. $M=10$ is employed in the test examples presented in this paper.
\begin{table*}[t]
\caption{Computational statistics -- all use tool length: $\bar{L}=10$ voxels.}%\vspace{-5pt}
\centering
\label{tab:CompStatistic}
\footnotesize
\begin{tabular}{l | c || r | r || c || r  r | r || r || r  r }
\hline
 & & Planning Space & & & \multicolumn{3}{c||}{Computational Time (Sec.)} & &\multicolumn{2}{c}{Support Voxels}\\
\cline{6-8} \cline{10-11}
Model  & Figure & Voxel Resolution & \# Solid Voxels  & $\Delta$-Value & Pre-Proc. & Nullification$^\ddag$ & Total & \# Operations & \# Pre-Proc. & \# Total \\
\hline \hline
GE-Bracket  & \ref{fig:teaser} & 100 $\times$ 60 $\times$ 34 & 16,403 & 10 & 0.63 & 11.46 (6.70) & 12.09 & 31,745 & 217 & 7,671\\
MBB Beam  & \ref{fig:MBBComp}(a) & 100 $\times$ 20 $\times$ 20 & 4,008 & 10 & 0.34 & 0.48 (0.39) & 0.82 & 8,908 & 734 & 2,450\\
Fertility$^\dag$  & \ref{fig:FertilityProgressResults} & 100 $\times$ 38 $\times$ 74 & 62,152 & 20 & 36.14 & 168.02 (129.14) & 204.16 & 80,774 & 635 & 9,311\\
Jennings Dog & \ref{fig:DogProgressiveHMResults} & 80 $\times$ 66 $\times$ 100 & 111,910 & 10 & 2.51 & 193.60 (125.90) & 196.11 & 124,310 & 693 & 6,200 \\
TPMS  & \ref{fig:TPMSResults} & 50 $\times$ 50 $\times$ 50 & 47,824 & 10 & 0.68 & 29.47 (25.49) & 30.15 & 65,868 & 1,428 & 9,022 \\

\hline
\end{tabular}
\begin{flushleft}
\footnotesize
$^\dag$~Note that a larger search range of $\Delta = 20$ is used for the Fertility model, which leads to fewer hybrid manufacturing operations while maintaining fast computation time. \\
$^\ddag$~The time reported in the bracket is for the total time used for stability evaluation.
\end{flushleft}
\end{table*}

\section{Results and Discussion}
The effectiveness and efficiency of our approach have been tested in both computational and physical experiments, which will be presented first. After that, the limitations and potential future work will be discussed at the end of this section.

\subsection{Computational Results}
Our method has been implemented in C++. All computational experiments are performed on a desktop PC with an Intel Core i9-13900K CPU (8 P-cores @ 5.8 GHz + 16 E-cores @ 4.3 GHz) with 128GB RAM, running Windows 10. 

\begin{figure}
\centering
\includegraphics[width=\linewidth]{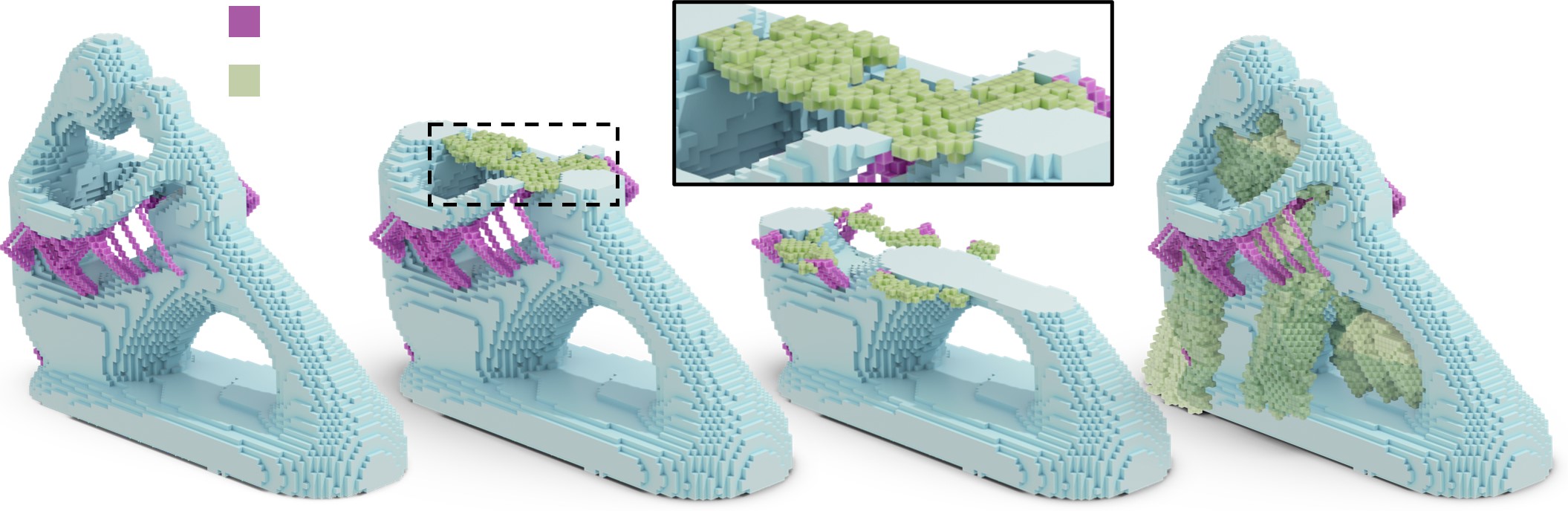}
\put(-200,74){\footnotesize \color{black}Added by Pre-Proc.}
\put(-200,65){\footnotesize \color{black}Added by Nullifi.}
\put(-232,3){\footnotesize \color{black}(a)}
\put(-172,3){\footnotesize \color{black}(b)}
\put(-112,3){\footnotesize \color{black}(c)}
\put(-52,3){\footnotesize \color{black}(d)}
\caption{The result of the hybrid manufacturing process planning on the Fertility model by using an SM tool with length $\bar{L} = 10$. From left to right -- (a) the result of pre-processing, (b, c) the progressive results of 1/3 \& 3/5 layers nullified, and (d) the model with all temporary support voxels displayed. The supports added in the pre-processing step are displayed in purple color. 
}\label{fig:FertilityProgressResults}
\end{figure}

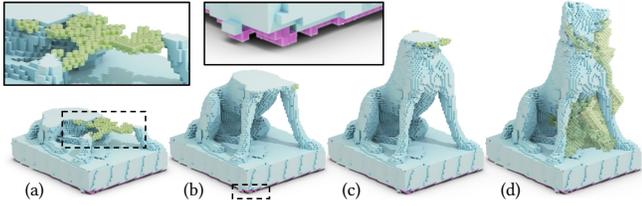
\begin{figure}
\centering
\includegraphics[width=\linewidth]{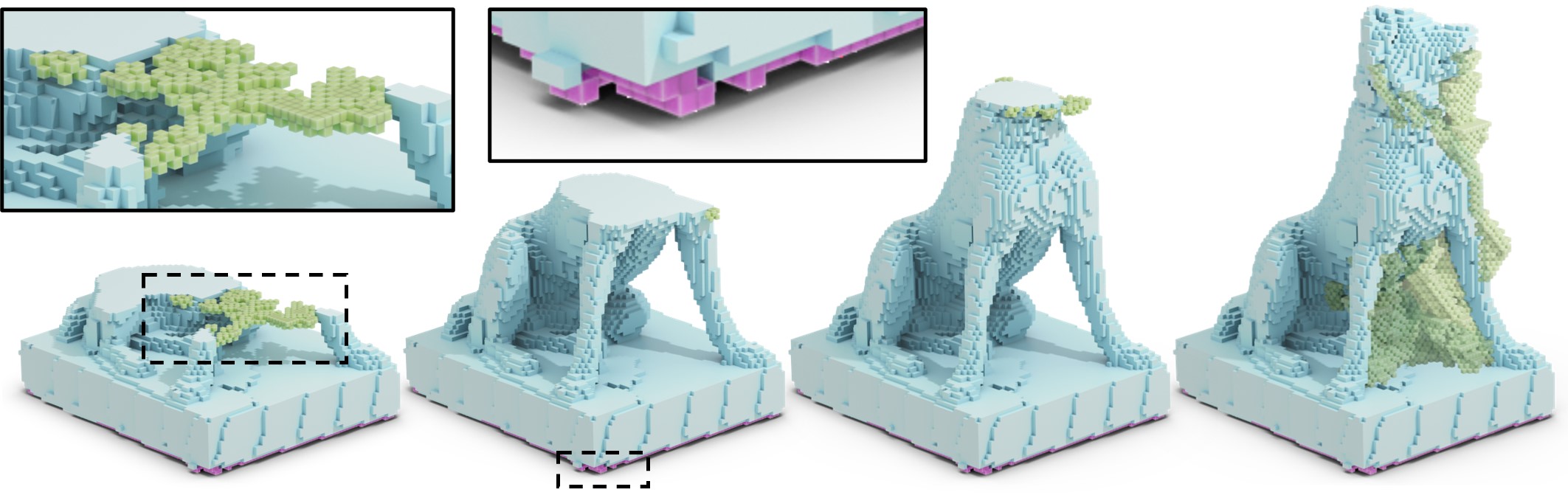}
\put(-235,3){\footnotesize \color{black}(a)}
\put(-175,3){\footnotesize \color{black}(b)}
\put(-115,3){\footnotesize \color{black}(c)}
\put(-55,3){\footnotesize \color{black}(d)}
\caption{The result of the hybrid manufacturing sequence determined by our algorithm on the Jennings Dog example -- (a-c) the progressive results of 1/4, 1/2 \& 3/4 height produced and (d) the model with all temporary support voxels displayed.
}\label{fig:DogProgressiveHMResults}
\end{figure}

\subsubsection{Examples and computational statistics}
We have evaluated our method on a variety of models with complex geometries. The first example is the GE-Bracket model generated through topology optimization (Fig.~\ref{fig:teaser}), a widely used benchmark in computational design and manufacturing. The model is voxelized at a resolution of $100 \times 60 \times 34$ for hybrid manufacturing planning, resulting in an input model containing 16.4k solid voxels. In the pre-processing step, 217 voxels are added as temporary supports. During the planning process, an additional 7.6k support voxels are introduced, resulting in a total of 31.7k hybrid manufacturing operations when using a tool length of $\bar{L} = 10$ voxels (equivalent to 12mm in physical units, assuming a voxel width of 1.2mm). 
The second model is the MBB beam shown in Fig.~\ref{fig:MBBComp}(a), which consists of 4.0k solid voxels. Using an SM tool with $\bar{L} = 10$ voxels, the result can be obtained in less than one second due to the small number of voxels involved.
The third example is the Fertility model, consisting of 62.2k solid voxels, shown in Fig.~\ref{fig:FertilityProgressResults}. For this case, we conservatively use the same shorter tool length of $\bar{L} = 10$ voxels. Compared to the GE-Bracket model, the Fertility model contains more solid voxels and thus requires a longer computation time. 
The forth example -- a Jennings Dog with more than 111.9k voxels as shown in Fig.~\ref{fig:DogProgressiveHMResults}. However, as fewer additional temporary supports (6.2k) are needed, both the total number of operations and the computational time are much smaller and shorter. Our planning algorithm gives a resultant sequence of 124.3k hybrid manufacturing operations. 

The final example is a triply periodic minimal surface (TPMS) structure as shown in Fig.~\ref{fig:TPMSResults}, which features an even more intricate topology (genus number 72). The model is voxelized at half resolution, yielding a comparable number of solid voxels (47.8k). Our algorithm determines a hybrid manufacturing sequence with 65.8k operations to fabricate this structure. All these examples present significant challenges for traditional AM or SM processes. However, feasible hybrid manufacturing plans can be computed efficiently with our method. Detailed computational statistics are summarized in Table~\ref{tab:CompStatistic}. It can be found that the stability check is still the most time-consuming step, even after applying the strategy of localized check. In some cases, up to 76.9\% of the total computation time for nullification is devoted to stability evaluation.

\begin{figure}
\centering
\includegraphics[width=\linewidth]{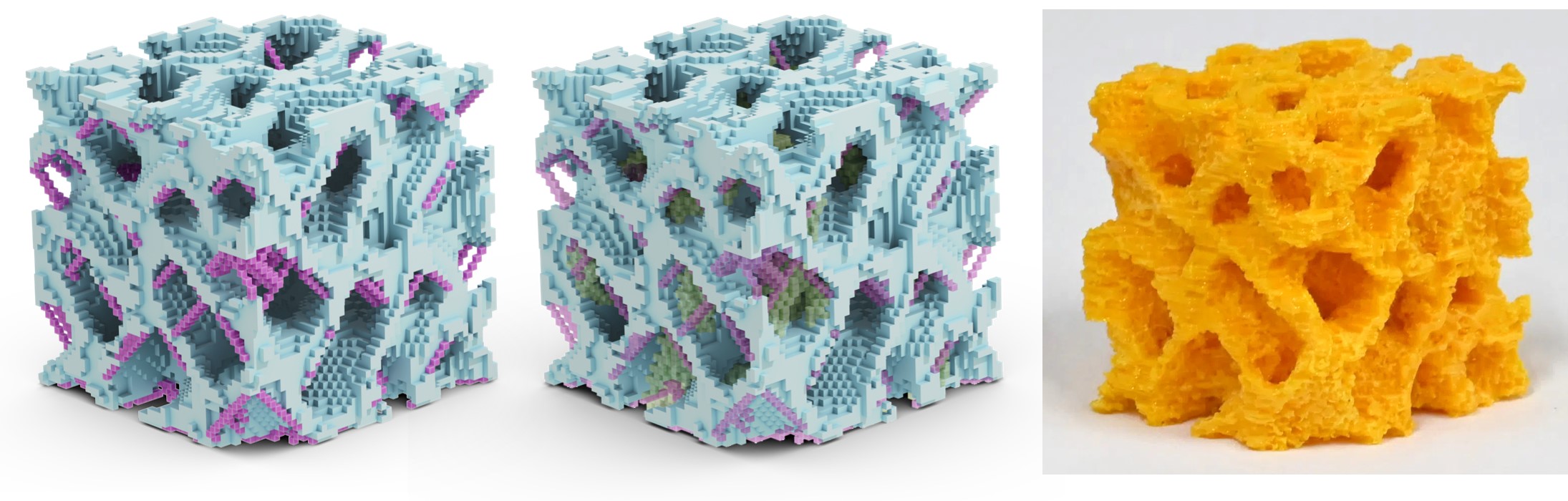}
\put(-243,8){\footnotesize \color{black}(a)}
\put(-165,8){\footnotesize \color{black}(b)}
\put(-90,8){\footnotesize \color{black}(c)}
\setlength{\abovecaptionskip}{4pt}
\caption{The result of a TPMS structure with complex topology (genus: 72) that can be automatically fabricated by our method -- (a) the model after pre-processing, (b) the model with all temporary support voxels visualized and (c) the physically fabricated model. 
}\label{fig:TPMSResults}
\end{figure}

\subsubsection{Study of SM tool length} 
After evaluating the overall performance of our algorithm on models with different topological complexity, we now study its behavior under different tool lengths for SM operations. We test a range of tool lengths, from the most conservative setting of $\bar{L}=2$ to the most aggressive choice of $\bar{L}=100$, including the practical setting of $\bar{L}=10$ used in our earlier experiments. These tests are conducted on two models -- the Fertility and the GE-Bracket -- with the results presented in Figs.~\ref{fig:toolLength_fertility} and \ref{fig:toolLength_bracket}. The results show that using very short SM tools leads to a large number of additional support voxels, which in turn increases computational time due to the extra effort required for the feasibility check.

We also quantify the number of voxels added during different phases of the pipeline -- specifically during the pre-processing and nullification steps. For both models, we observe that longer SM tools significantly reduce the number of voxels added during the nullification step, although they may increase computation time in certain cases, such as when changing from $\bar{L}=10$ to $\bar{L}=30$. 

\begin{figure}
    \centering
    \includegraphics[width=1\linewidth]{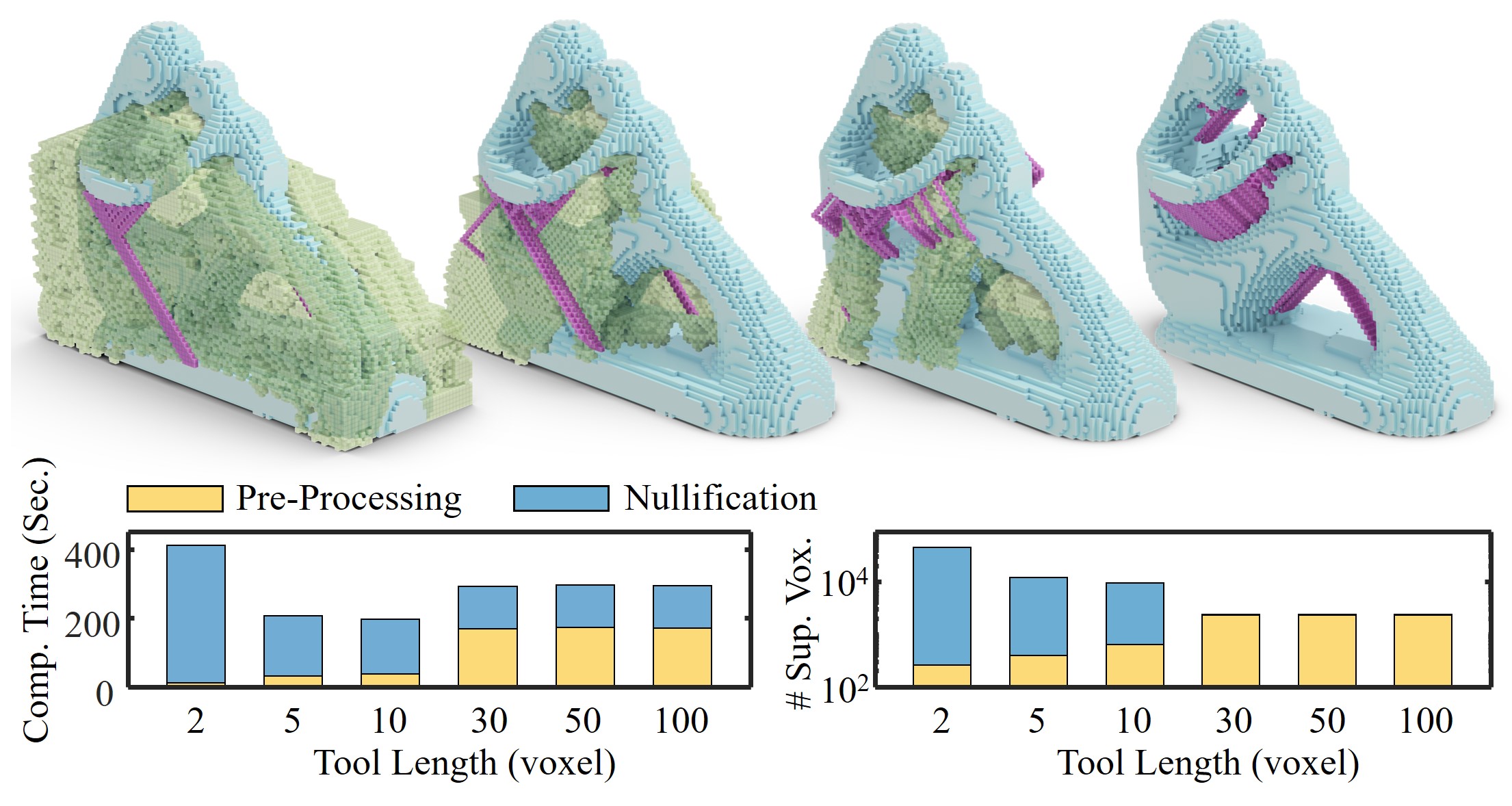}
    % \put(-240,189){\footnotesize \color{black} \textbf{Without Pre-Processing} }
    % \put(-240,127){\footnotesize \color{black}\textbf{With Pre-Processing}}
    \put(-205,118){\footnotesize \color{black}$\Bar{L}=2$ }
    \put(-140,118){\footnotesize \color{black}$\Bar{L}=5$ }
    \put(-85,118){\footnotesize \color{black}$\Bar{L}=10$ }
    \put(-34,118){\footnotesize \color{black}$\Bar{L}=100$ }
\setlength{\abovecaptionskip}{5pt}
\caption{Results using different tool lengths on the Fertility model (voxel Res.: $100 \times 38 \times 74$). Charts showing the computation time and the number of support voxels are provided.
}\label{fig:toolLength_fertility}
\end{figure}

\begin{figure}
    \centering
    \includegraphics[width=1\linewidth]{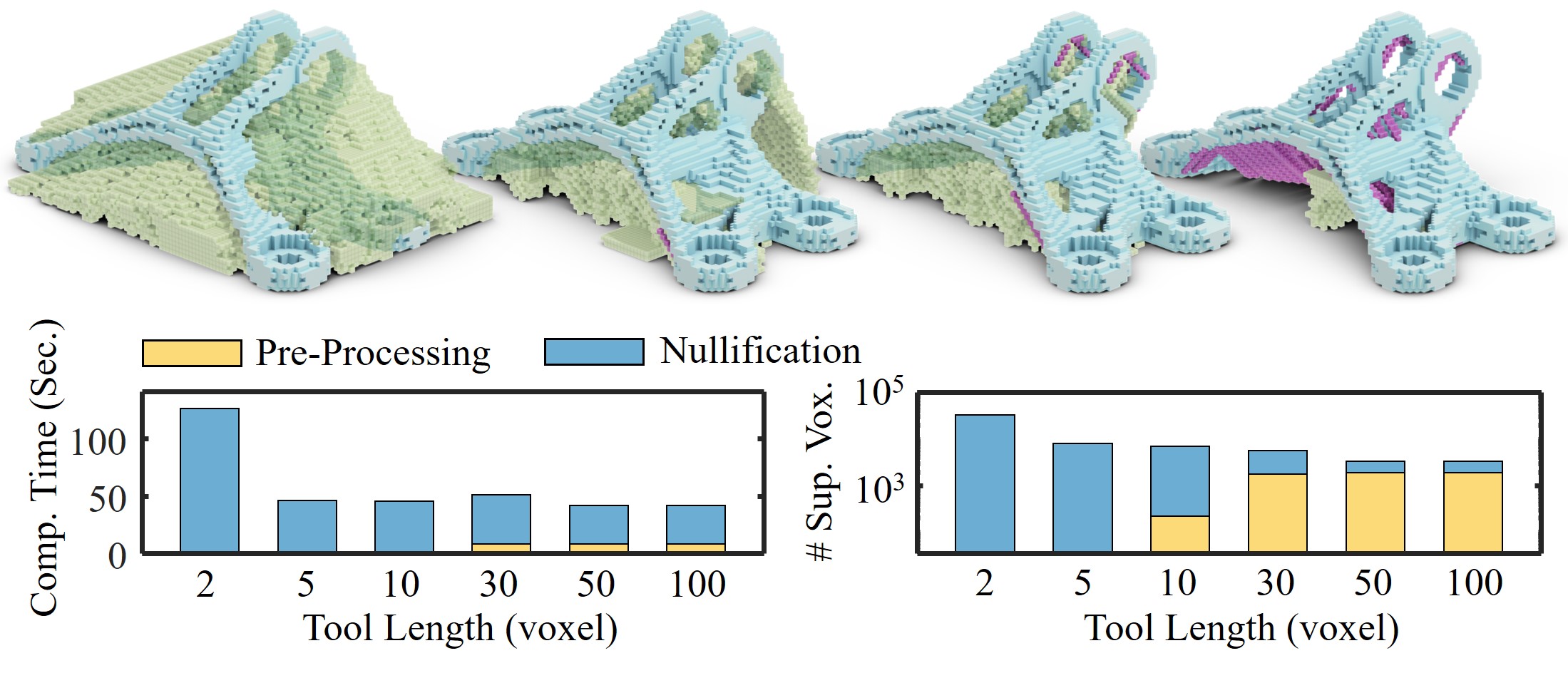}
    \put(-236,95){\footnotesize \color{black}$\Bar{L}=2$ }
    \put(-168,95){\footnotesize \color{black}$\Bar{L}=5$ }
    \put(-112,95){\footnotesize \color{black}$\Bar{L}=10$ }
    \put(-60,95){\footnotesize \color{black}$\Bar{L}=100$ }
    \setlength{\abovecaptionskip}{5pt}
    \caption{Results using different tool lengths on the GE-Bracket model (voxel Res.: $100 \times 60 \times 34$). Charts showing the computation time and the number of support voxels are provided.
    }\label{fig:toolLength_bracket}
\end{figure}

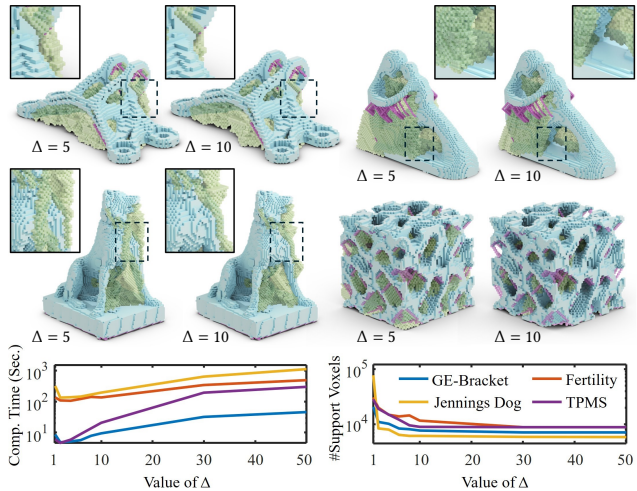
\begin{figure}
\centering
\includegraphics[width=1\linewidth]{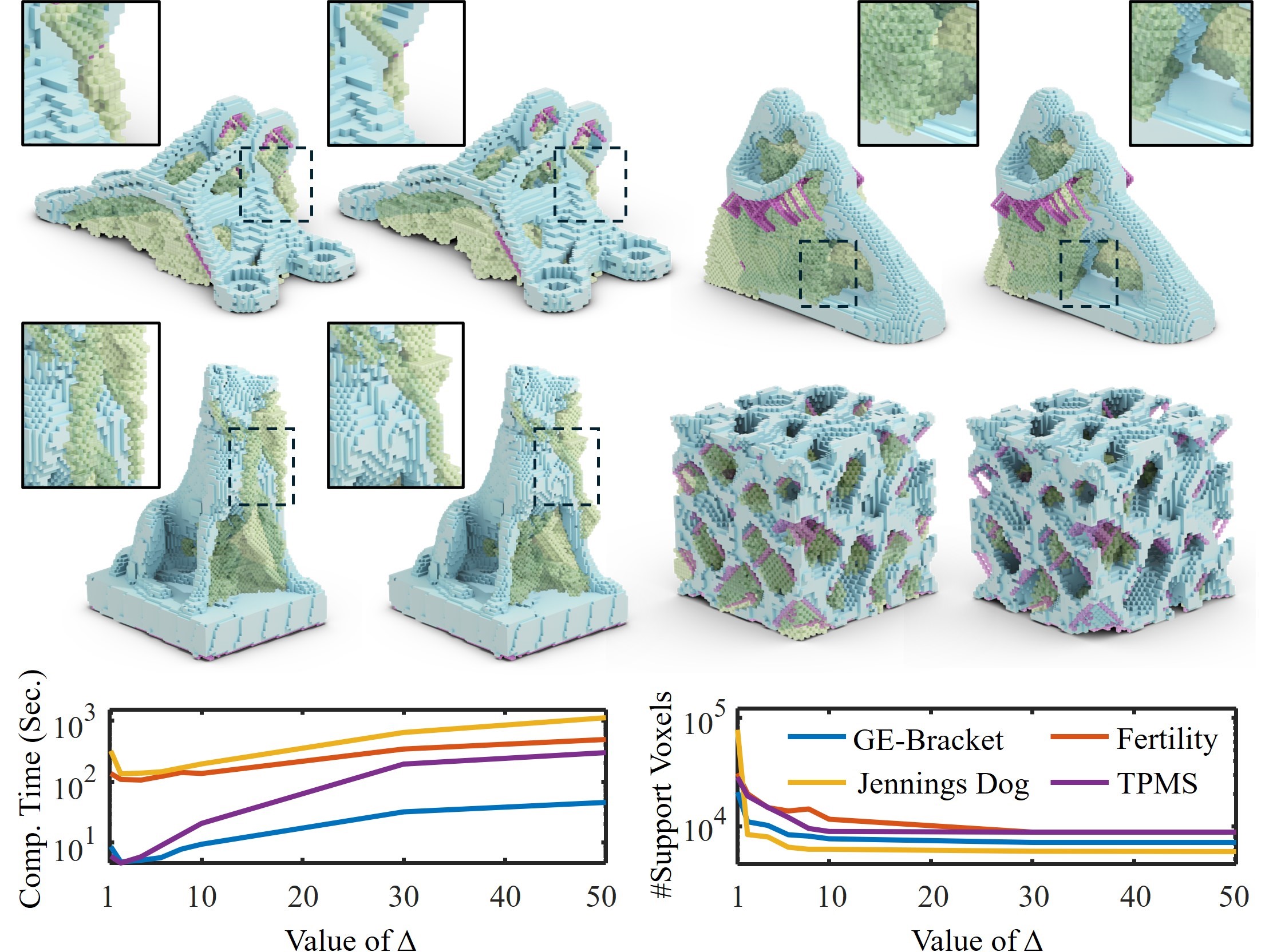}
\put(-231,127){\scriptsize \color{black}$\Delta=5$ }
\put(-173,127){\scriptsize \color{black}$\Delta=10$ }
\put(-108,116){\scriptsize \color{black}$\Delta=5$ }
\put(-56,116){\scriptsize \color{black}$\Delta=10$ }
\put(-231,55){\scriptsize \color{black}$\Delta=5$ }
\put(-174,55){\scriptsize \color{black}$\Delta=10$ }
\put(-108,55){\scriptsize \color{black}$\Delta=5$ }
\put(-56,55){\scriptsize \color{black}$\Delta=10$ }
\caption{Study for the influence of the search range $\Delta$ on the computing time and the number of additional support voxels.}\label{fig:stabilityCheckSize}
\end{figure}

\subsubsection{Search range for stability check}
Another important study concerns the influence of the search range used in the localized stability check -- i.e., the parameter $\Delta$ defined in Def.~\ref{def:DeltaNeighbor}. We conduct tests on four models, with the results shown in Fig.~\ref{fig:stabilityCheckSize}. As expected, a larger search range generally increases computation time. On the other hand, increasing the value of $\Delta$ helps reduce the number of `fake' unstable cases, thereby minimizing the number of unnecessary support voxels added. To keep a balance between computational efficiency and solution quality, we use $\Delta = 10$ or $20$ for the examples presented in this paper.

\subsubsection{Validation of scalability}
We evaluate our method on the GE-Bracket and Fertility models at varying resolutions, using a fixed tool length of $\bar{L} = 10$ voxels and a local search range of $\Delta = 10$. The results are presented in Fig.~\ref{fig:diffResolution}. Thanks to the scalable design of our algorithm, it performs well even on models with up to 979k solid voxels -- see the Fertility model at the resolution of $250 \times 95 \times 186$. The result has more than 1.26M HM operations. As the resolution increases -- meaning the model size grows relative to the fixed tool length -- it becomes harder for the SM tool to access interior regions, resulting in fewer support structures being added during the pre-processing step.

\begin{figure}
\centering
\includegraphics[width=\linewidth]{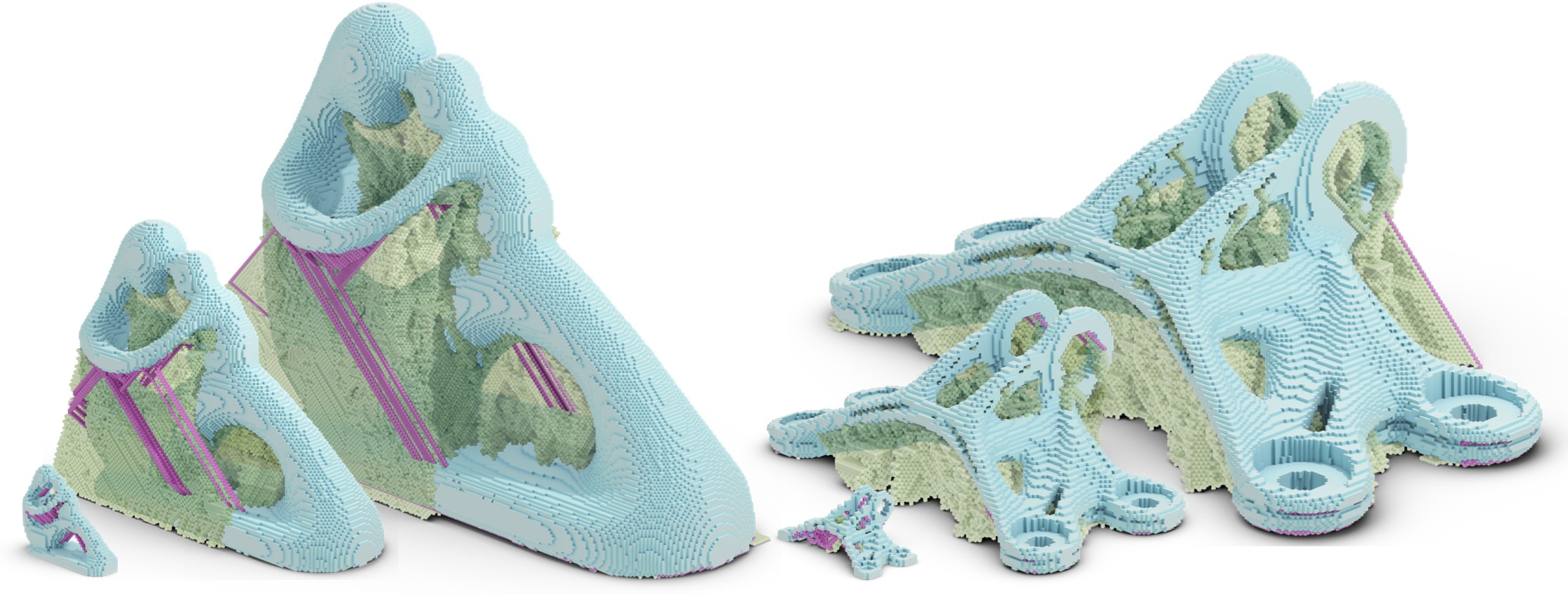}\\
%\vspace{5pt}
\footnotesize
   \begin{tabular}{r|r|r|r|r|r}
    \hline
        & \multicolumn{2}{c|}{Voxel \#} & \multicolumn{2}{c|}{Comp. Time (Sec.)} & Total \\
        \cline{2-5}
    Fertility  & Solid  & Support  & Pre-proc. & Nullification &  Operation \# \\
    \hline
    \hline
    50 $\times$ 18 $\times$ 37 & 7556 & 437 & 0.08 & 1.50 & 8,430 \\
    150 $\times$ 57 $\times$ 111 & 211,002 & 41,105 & 642.46 & 787.90 & 293,212 \\
    250 $\times$ 95 $\times$ 186 & 979,465 & 139,361 & 12107.50 & 13223.2 & 1,258,187 \\
    \hline
    \hline
        & \multicolumn{2}{c|}{Voxel \#} & \multicolumn{2}{c|}{Comp. Time (Sec.)} & Total \\
        \cline{2-5}
    GE-Bracket & Solid  & Support & Pre-proc. & Nullifi. &  Operation \# \\
    \hline
    \hline
    50 $\times$ 30 $\times$ 17 & 2,009 & 548 & 0.02 & 0.25 & 3,105 \\
    150 $\times$ 90 $\times$ 52 & 55,697 & 25,735 & 31.29 & 238.03 & 107,167 \\
    250 $\times$ 151 $\times$ 87 & 258,581 & 101,330 & 520.38 & 4061.32 & 461,241 \\
    \hline
    \end{tabular}
\caption{Results of the Fertility and GE-Bracket models in different resolutions while using the same SM tool length as $\bar{L}=10$. %\yongxue{updated}%\charlie{please change support from \% to voxel \#.}\yongxue{changed}
}\label{fig:diffResolution}
\end{figure}

\subsubsection{Verification of completeness} 
To verify the completeness of our method, we test our nullification algorithm on 100 models randomly selected from the Thingi10K dataset~\cite{Thingi10K}. The hybrid manufacturing sequences for all models can be successfully generated, and the results are shown in Fig.~\ref{fig:Thingi10K}. 

\subsection{Physical Fabrication}
\begin{figure}
\centering
\includegraphics[width=\linewidth]{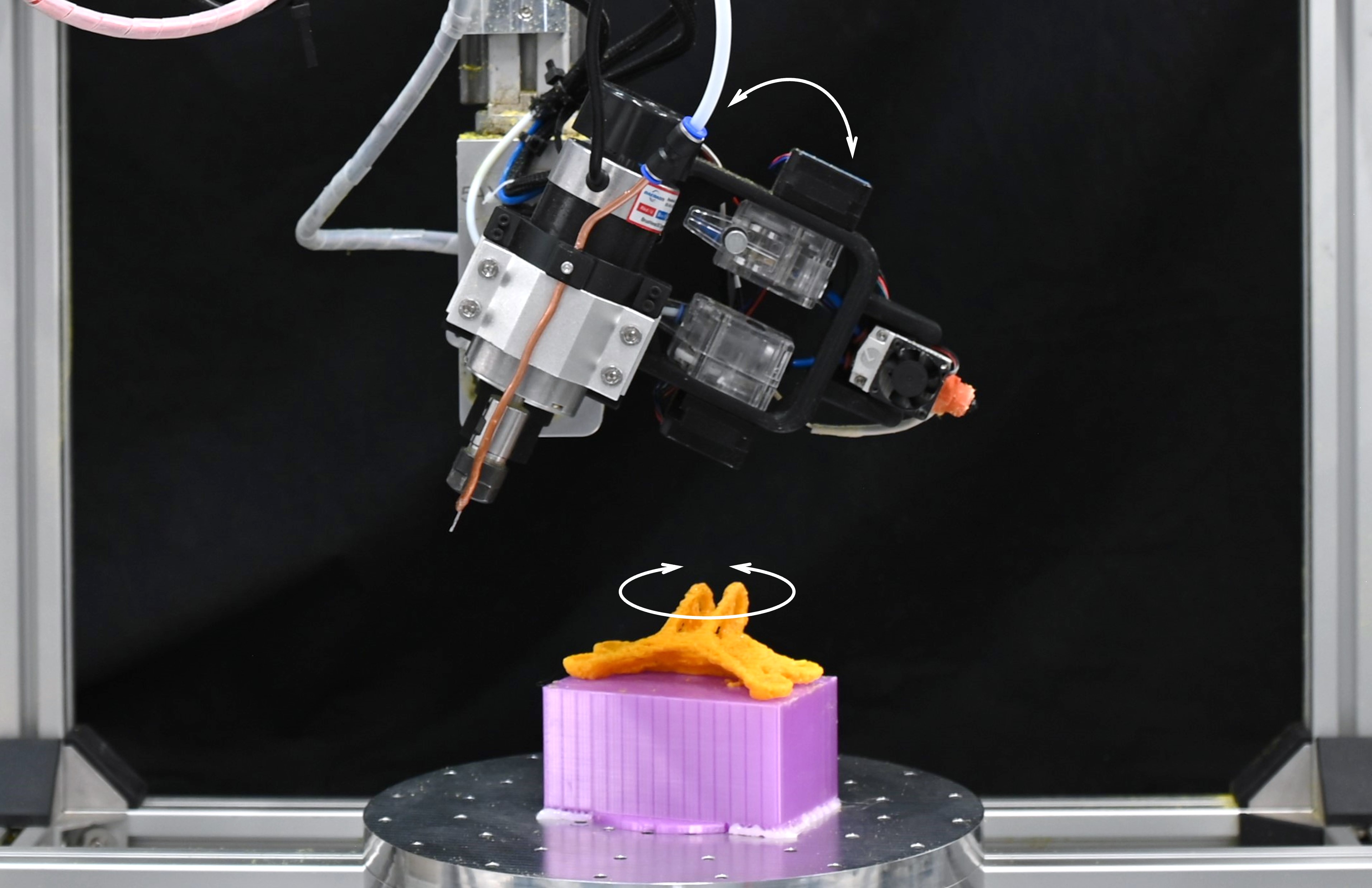}
\put(-87,140){\footnotesize \color{white}B-axis}
\put(-97,57){\footnotesize \color{white}C-axis}
\caption{The hybrid manufacturing hardware employed in physical verification, where the rotations around the B-axis and C-axis are illustrated. 
}\label{fig:Hardware}
\end{figure}

\begin{figure}
\centering
\includegraphics[width=\linewidth]{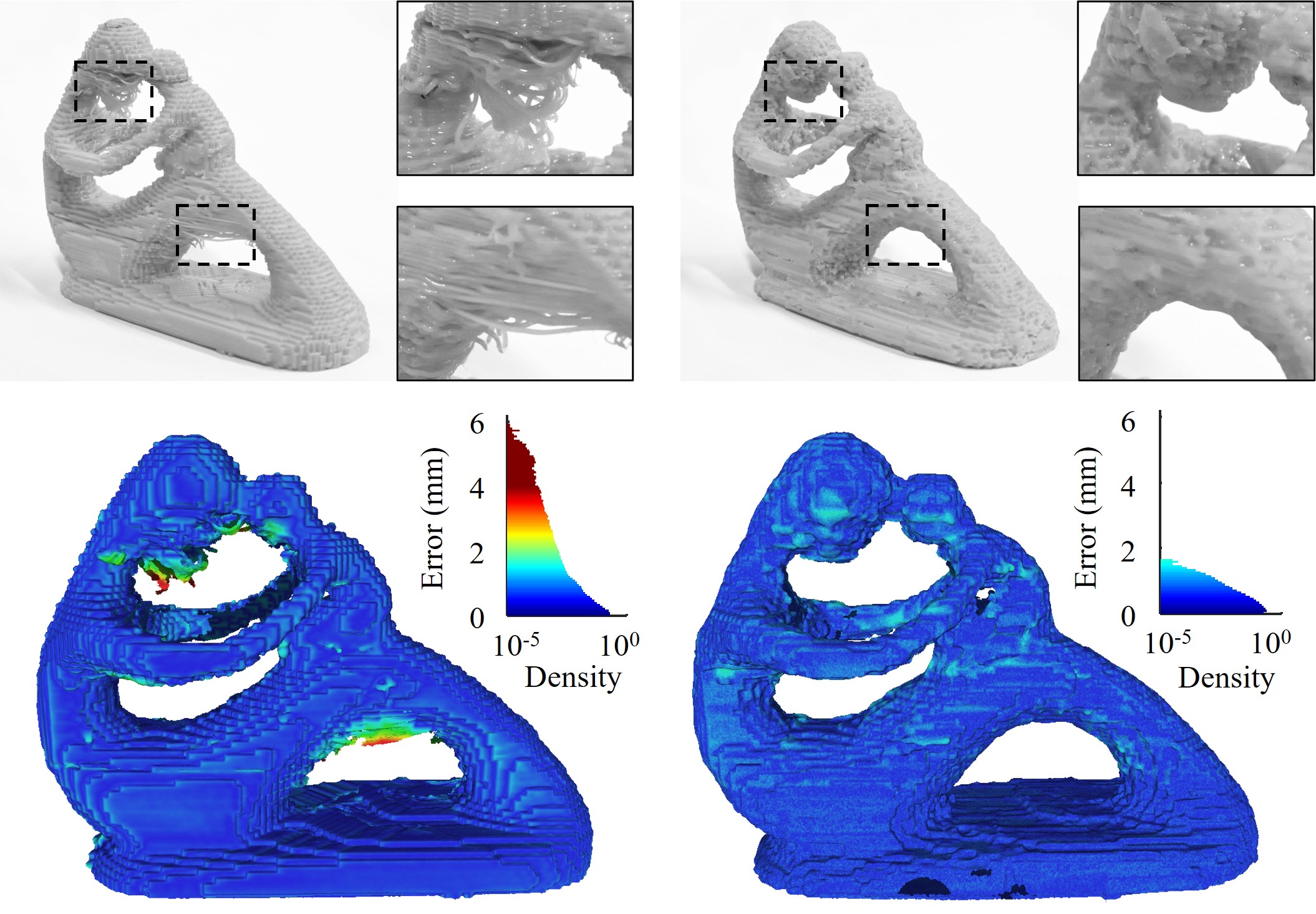}
\put(-240,3){\footnotesize \color{black}(a)}
\put(-118,3){\footnotesize \color{black}(b)}
\caption{Quantitative comparison of the Fertility model, enhanced with machinable support structures generated by our pre-processing step, fabricated using (a) AM-only operations (with pre-added supports removed by SM operations at the end) and (b) the hybrid operations planned by our method. The top row shows the physical fabrication results, while the bottom row visualizes the geometric errors of the 3D scanned models relative to the target shape using color maps.}\label{fig:FertilityPhysicalResults}
\end{figure}

\subsubsection{Hardware}
We validated the hybrid manufacturing sequences generated by our algorithm on a machine modified from a 5-axis CNC system (5XM600XL from 5AxisMaker) as shown in Fig.~\ref{fig:Hardware}. A newly designed mounting bracket was developed to position the printhead (for AM) and the spindle (for SM) in an orthogonal configuration. This setup enables automatic switching between AM and SM operations and allows the SM tool to be oriented either vertically or horizontally by the rotation around the B-axis. Additionally, by rotating the machine’s C-axis (located at the base) to four orthogonal positions, four distinct horizontal SM tool orientations relative to the workpiece can be achieved.

For the AM operations, we use a Fused Deposition Modeling (FDM) extruder equipped with a 0.6mm diameter nozzle, operating at an extrusion rate of 400mm/min. The SM capability is provided by a Baedalus CNC brushless spindle (500W, 48V DC, with a maximum speed of 12,000 RPM). A single-flute end mill with 1.2mm diameter cutter and 12mm tool length (i.e., $\bar{L} = 10$) is employed for machining operations. The following cutting parameters are used: spindle speed of 6,000RPM and feedrate of 900mm/min.

System control and synchronization are handled through an integrated architecture. A Duet 3 MainBoard 6HC serves as the control board, managing both the filament feed mechanism of the FDM extruder and the on/off state of the SM spindle. A host PC oversees the overall operation. Communication between the host PC and the CNC machine’s native motion controller is established via a serial connection, while the PC interfaces with the control board over Ethernet. This setup enables synchronized coordination between the machine’s kinematic movements and the HM process.

\begin{figure}
\centering
\includegraphics[width=.9\linewidth]{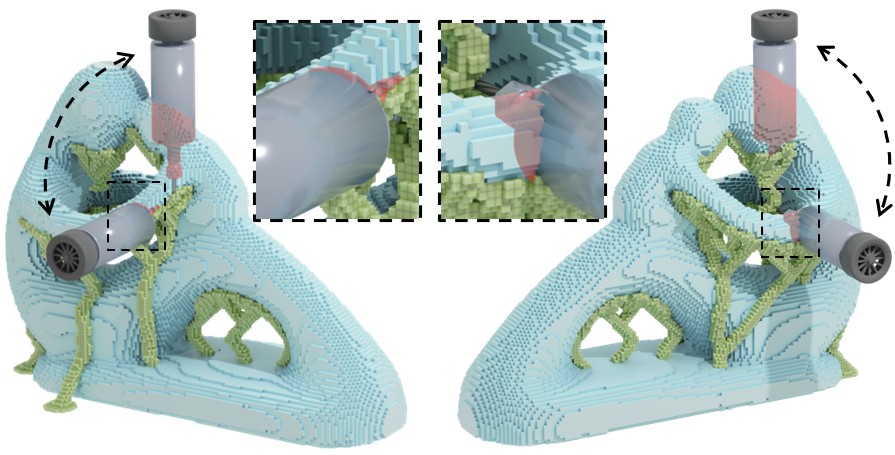}
\put(-218,3){\footnotesize \color{black}(Front View)}
\put(-43,3){\footnotesize \color{black}(Back View)}
\caption{Collision occurs (with interference regions highlighted in red) when using a 10mm SM tool to remove the AM support generated by Adobe MeshMixer for the Fertility model -- see the illustration in different views.
}\label{fig:fertility_AMSM}
\end{figure}

\subsubsection{Fabrication and verification}
We physically fabricated several of the previously shown examples using our hybrid machine, employing PETG filament for the AM operations as the material is resistant to high temperatures and has better milling performance. The fabrication result of the GE-Bracket has been presented in Fig.~\ref{fig:teaser}, demonstrating the effectiveness of our approach. For comparison, we also fabricated the same model using only AM operations, augmented with removable support structures generated in the pre-processing step. 
To ensure a fair comparison, the pre-added support structures are removed using SM operations after the printing-only stage. As shown in the 3D scanning results, the AM-only fabrication fails in regions with large overhangs, whereas the interleaved AM and SM operations in our method produce successful outcomes. 
A similar comparison is provided for the Fertility model, where geometric deviations are visualized using color maps derived from 3D scans (see Fig.~\ref{fig:FertilityPhysicalResults}). As illustrated in Fig.~\ref{fig:fertility_AMSM}, the supports added by a commercial software -- Adobe MeshMixer cannot be reached by our SM cutter with the tool length as 10mm.

We fabricated both MBB beam models shown in Fig.~\ref{fig:MBBComp}: the first was produced following the HM sequence generated by our method, while the second was directly 3D printed. The fabricated models are shown in Fig.~\ref{fig:MBB_results}. To further verify structural stiffness, both specimens were subjected to 3-point bending tests taken on the INSTRON 3344 1kN tensile testing machine. As shown in Fig.~\ref{fig:MBB_results}(b), the MBB beam computed by TO without the self-support constraint and fabricated via our HM sequence sustained a 30.51\% improvement of stiffness.
Additionally, we also fabricated a TPMS model with complex topology with its result shown on the right side of Fig.~\ref{fig:TPMSResults}. While our algorithm is demonstrated on a polymer-based hybrid machine, it is general and applicable across a wide range of materials and hardware setups.

\subsection{Discussion}
While our algorithm focuses on ensuring the manufacturability of models with arbitrary shapes, it only implicitly encourages continuous AM operations through the nullification algorithm and the post-processing stage of toolpath generation. Minimizing the total number of tool switches is not a primary objective in our sequence planning. Although this is generally acceptable for machines equipped with automatic tool-changing capabilities, frequent switching may introduce mechanical errors, potentially compromising the precision of the fabricated model in hybrid manufacturing. Such errors can be mitigated through precise calibration and the use of a high-stiffness switching mechanism, but they cannot be entirely eliminated.

\begin{figure}
\centering
    \includegraphics[width=\linewidth]{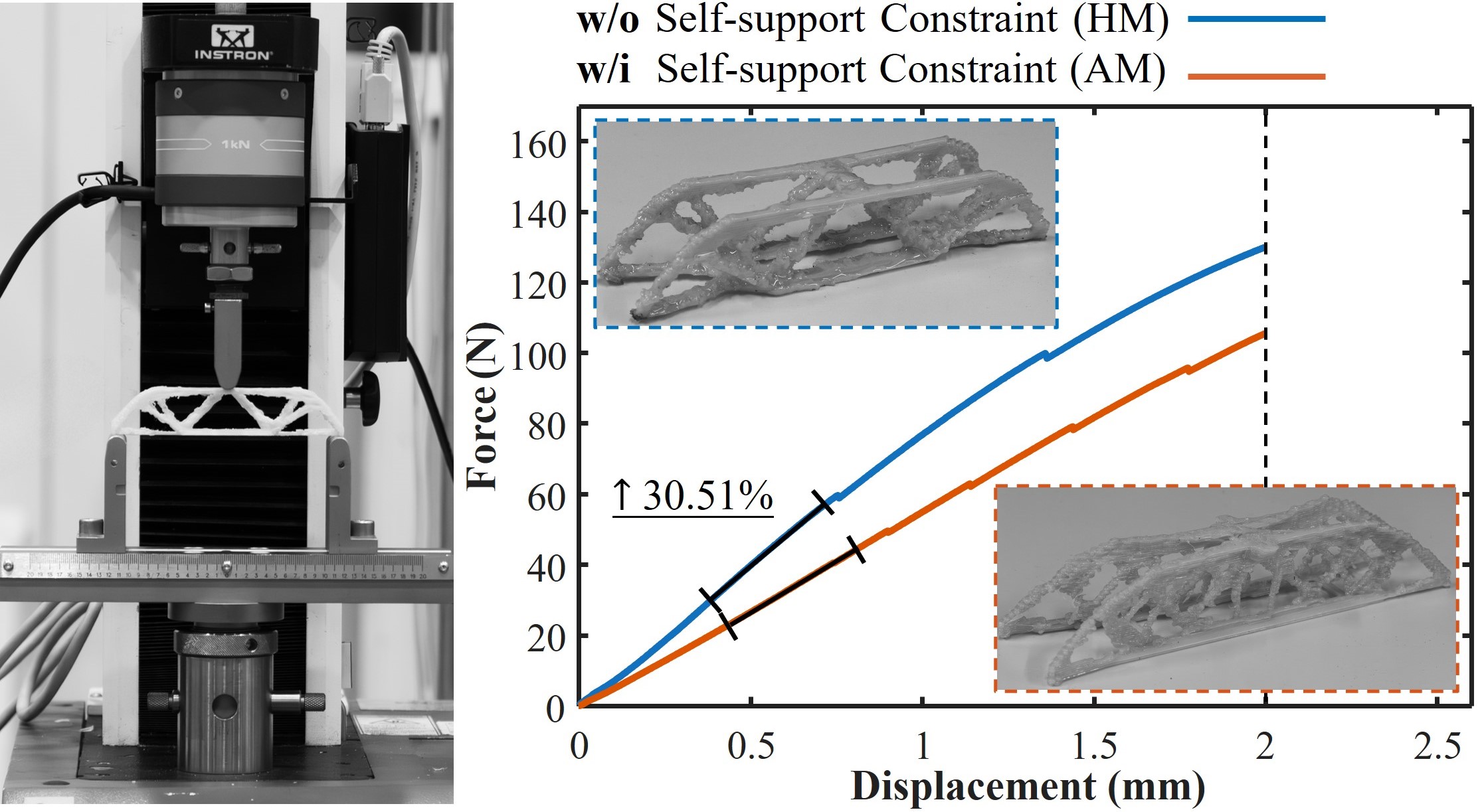}
    %\put(-240,85){\footnotesize \color{black}(a)}
    %\put(-240,3){\footnotesize \color{black}(b)}
    \put(-242,3){\footnotesize \color{black}(a)}
    \put(-165,3){\footnotesize \color{black}(b)}
    %\put(-83,78){\footnotesize \color{black}(c)}
    %\put(-83,23){\footnotesize \color{black}(d)}
    \caption{Three-point bending tests (a) taken on the physical specimens of MBB beams determined by TO with (fabricated by AM with weight: 8.2g) vs. without self-support constraints (made by our HM with weight: 8.1g). (b) The force–displacement curves show 30.51\% larger stiffness on the MBB beam made by our HM method. Note that the mechanical tests are limited to a displacement range of 0–2mm, as this range better matches the loading conditions used in the topology optimization.
    }\label{fig:MBB_results}
\end{figure}

One limitation of our approach is that it does not account for material deformation caused by gravity, as the definition of structural stability in this paper assumes the use of highly stiff materials during hybrid manufacturing. However, in practice, structures under fabrication may experience gravitational deformation -- an issue addressed in prior work on additive manufacturing \cite{Huang2024SIGA, Wang2020SMO, Stava2012SIG}. Incorporating such deformation effects into our stability analysis is an important future research.

Another limitation of our current approach is that it does not explicitly address the removal of SM chips within closed cavities -- i.e., geometries that have been explored in previous work such as~\cite{Romain2013Tog} to design self-balanced models suitable for 3D printing. While such cavities can sometimes be optimized to be self-supporting~\cite{Wang2018TVCG}, the issue of trapped material remains a practical concern in hybrid manufacturing. In our fabrication experiments, we apply a practical engineering approach -- a high-pressure air blower is used to forcibly remove residual material before AM operation closing the cavity. While this approach is effective in many cases, it is not guaranteed to work for all geometries -- especially those with deep or highly tortuous internal voids, where the blowing method may not be able to remove the chips of SM operations completely. We consider this as an area for future improvement in HM process planning.

When claiming that any shape can be fabricated, we disregard the orientation of material properties, as this work considers only isotropic materials. However, for complex composites with strong anisotropic behavior (e.g.,~\cite{Liu2025SIG}), integrated optimization of both design and manufacturing objectives is still required. This, however, is beyond the scope of this paper.

Lastly, for ease of implementation, our physical fabrication was carried out on a machine equipped with a spindle of relatively low stiffness, which limited our experiments to polymer materials. Extending the application of our HM process planning algorithm to metal-based HM systems presents an exciting avenue for future work. Furthermore, replacing the CNC tool in our setup with a laser cutter could offer the advantage of processing very deep or hard-to-reach cavities. Nevertheless, the practical use of laser cutting in such contexts would still need to account for limitations such as beam divergence, focal depth, and material-specific ablation characteristics.

\begin{figure*}
\centering
\includegraphics[width=\linewidth]{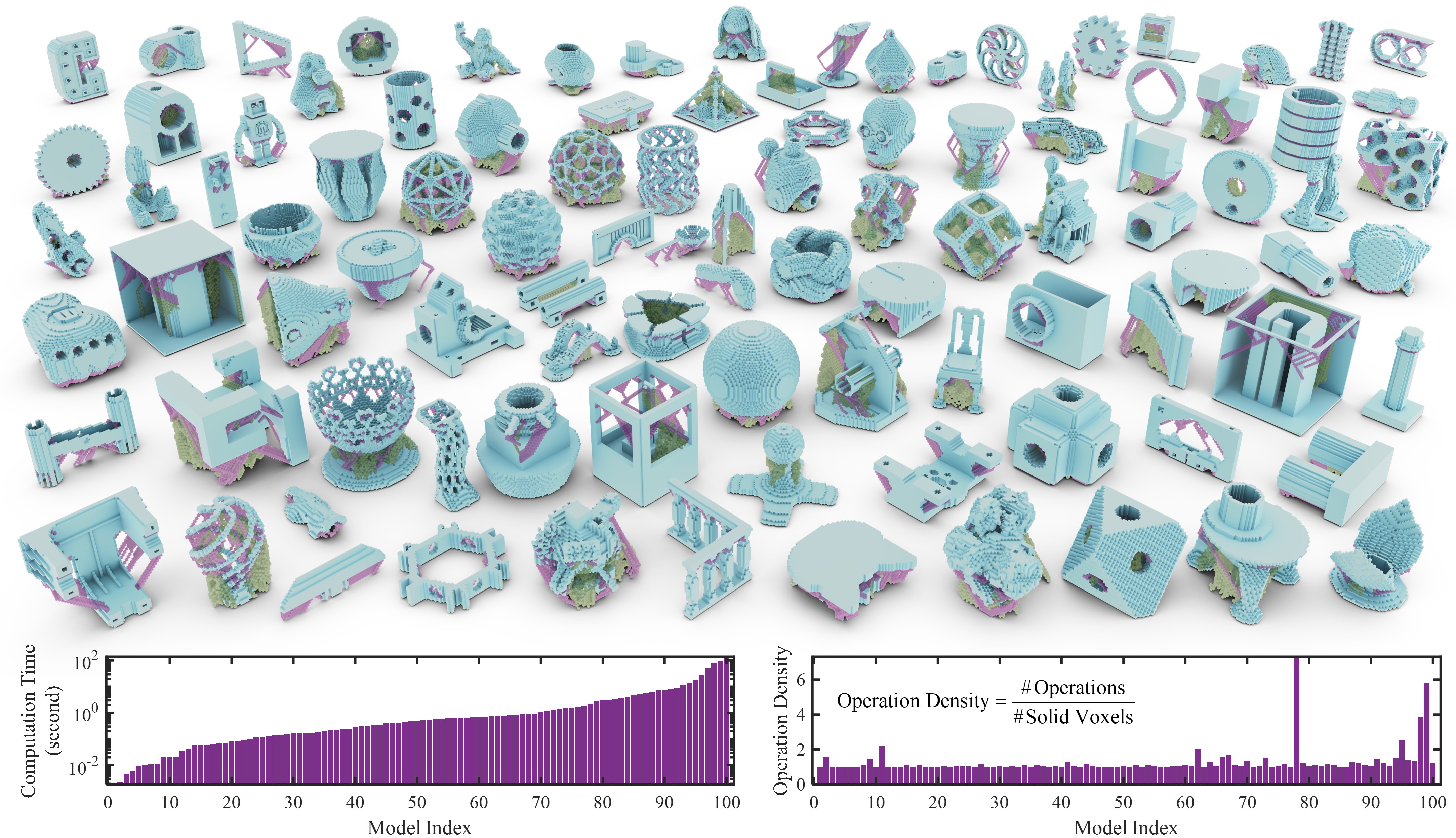}\\
%\vspace{-5pt}
\caption{Results of our method on 100 models randomly selected from the Thingi10K dataset. Support voxels added during the nullification step are shown in green, while those added during the pre-processing step are shown in purple. Computational time and operation density across all 100 models are also reported, where operation density denotes the ratio of total operations to the number of solid voxels.
}\label{fig:Thingi10K}
\end{figure*}

\section{Conclusion}
In this paper, we propose a nullification algorithm to compute interleaved AM and SM operations for fabricating models of arbitrary shapes. The algorithm is built upon two fundamental operations -- accretion and erosion -- which are the inverse of AM and SM operations respectively. During the nullification process, these operations are applied in a coordinated manner to gradually reduce the model to an empty set as null, while ensuring both manufacturability and the structural stability of all intermediate shapes. We provide a theoretical proof that the nullification process is always completable for any input model represented as a voxel set. Furthermore, we develop a scalable implementation that is able to handle models with up to 0.98M solid voxels. As the first approach to theoretically guarantee the feasibility of hybrid manufacturing for arbitrary geometries, our algorithm has been validated both computationally by a variety of models and experimentally through physical fabrication.

\begin{acks}
This project was supported by the Chair Professorship Fund at the University of Manchester and the UK Engineering and Physical Sciences Research Council (EPSRC) Fellowship Grant (Ref: EP/X032213/1). The authors also gratefully acknowledge the technical support provided by 5AXISWORKS Ltd.
\end{acks}

\bibliographystyle{ACM-Reference-Format}
\bibliography{reference}
\end{document}